\documentclass[a4paper,USenglish,cleveref,autoref,numberwithinsect]{lipics-v2019}

\usepackage{mathtools,bussproofs} 

\def\mytitle{A Symmetric Lambda-Calculus Corresponding to the~Negation-Free Bilateral Natural Deduction}
\def\myshorttitle{A Symmetric Lambda-Calculus Corresponding to a Bilateral Natural Deduction}
\def\myabstract{
  Filinski constructed a symmetric lambda-calculus consisting of
  expressions and continuations which are symmetric, and functions
  which have duality. In his calculus, functions can be encoded to
  expressions and continuations using primitive operators. That is,
  the duality of functions is not derived in the calculus but adopted
  as a principle of the calculus.  In this paper, we propose a simple
  symmetric lambda-calculus corresponding to the negation-free natural
  deduction based bilateralism in proof-theoretic semantics. In our
  calculus, continuation types are represented as not negations of
  formulae but formulae with negative polarity. Function types are
  represented as the implication and but-not connectives in
  intuitionistic and paraconsistent logics, respectively.  Our
  calculus is not only simple but also powerful as it includes a
  call-value calculus corresponding to the call-by-value dual calculus
  invented by Wadler.  We show that mutual transformations between
  expressions and continuations are definable in our calculus to
  justify the duality of functions. We also show that every typable
  function has dual types. Thus, the duality of function is derived
  from bilateralism.
}

\def\mykeywords{symmetric lambda-calculus, formulae-as-types, duality,
  bilateralism, natural deduction, proof-theoretic semantics,
  but-not connective, continuation, call-by-value}

\title{\mytitle}
\titlerunning{\myshorttitle}

\author{Tatsuya Abe}{Software Technology and Artificial Intelligence Research Laboratory, Chiba Institute of Technology, 2-17-1 Tsudanuma, Narashino, Chiba, 275-0016, Japan}{abet@stair.center}{}{}

\author{Daisuke Kimura}{Department of Information Science, Toho University, 2-2-1 Miyama, Funabashi, Chiba, 274-8510, Japan}{kmr@is.sci.toho-u.ac.jp}{}{}
\authorrunning{Tatsuya Abe and Daisuke Kimura} 


\category{}

\relatedversion{}

\supplement{}

\acknowledgements{The author thanks Yosuke Fukuda, Tasuku Hiraishi,
  Kentaro Kikuchi, and Takeshi Tsukada for the fruitful discussions,
  which clarified contributions of the present paper.
}

\nolinenumbers 

\hideLIPIcs  

\EventEditors{John Q. Open and Joan R. Access}
\EventNoEds{2}
\EventLongTitle{42nd Conference on Very Important Topics (CVIT 2016)}
\EventShortTitle{CVIT 2016}
\EventAcronym{CVIT}
\EventYear{2016}
\EventDate{December 24--27, 2016}
\EventLocation{Little Whinging, United Kingdom}
\EventLogo{}
\SeriesVolume{42}
\ArticleNo{23}

\begin{document}
\maketitle
\begin{abstract}
\myabstract
\end{abstract}


\allowdisplaybreaks

\def\seeappendix{
\begin{proof}
See Appendix.
\end{proof}
}

\newcommand{\emphtwo}[1]{\textcolor{blue}{#1}}
\newcommand{\emphgray}[1]{\textcolor{gray}{#1}}

\newcommand{\OK}{\mathord{\checkmark}}
\newcommand{\NA}{\mathord{\hspace{.6pt}\times\hspace{.6pt}}}

\newcommand{\ifthenelsefi}[4]{\ttif^{#1} \: {#2}\mathord{?}{#3}\mathord{:}{#4}}
\newcommand{\whiledoend}[3]{\ttwhile^{#1} \: {#2}\mathord{?}{#3}}
\newcommand{\letin}[2]{\ttlet \: #1 \: \texttt{in} \: #2}
\newcommand{\caseterm}[5]{\mathrm{case}(#1 , #2 . #3 , #4 . #5)}
\newcommand{\inl}[1]{\mathrm{inl}(#1)}
\newcommand{\inr}[1]{\mathrm{inr}(#1)}
\newcommand{\rec}[3]{\texttt{rec(}{#1\texttt{,}#2\texttt{,}#3}\texttt{)}}

\newcommand{\lamabs}[2]{\lambda #1 . #2}
\newcommand{\app}[2]{#1 #2}
\newcommand{\proj}[2]{\pi_{#1}(#2)}
\newcommand{\fst}[1]{\proj{0}{#1}}
\newcommand{\snd}[1]{\proj{1}{#1}}

\newcommand{\pair}[2]{\mathord{\langle #1,#2 \rangle}}
\newcommand{\triple}[3]{\mathord{\langle #1,#2,#3 \rangle}}
\newcommand{\negop}[1]{\operatorname{\neg}{#1}}
\newcommand{\existsop}[1]{\operatorname{\exists}{#1}}
\newcommand{\forallop}[1]{\operatorname{\forall}{#1}}
\newcommand{\defined}[1]{\texttt{df}(#1)}
\newcommand{\fv}[1]{\operatorname{fv}(#1)}
\newcommand{\ev}[2]{[\![#1]\!]_{#2}}
\newcommand{\dom}[1]{\operatorname{dom} #1}
\newcommand{\tran}[1]{\mathrel{{#1}^{\mathord+}}}
\newcommand{\cls}[1]{\mathrel{{#1}^{\ast}}}
\newcommand{\lub}[1]{\lceil #1 \rceil}
\newcommand{\glb}[1]{\lfloor #1 \rfloor}
\newcommand{\intvl}[2]{[#1,#2)}
\newcommand{\intvlext}[4]{{\intvl{#3}{#4}}_{\langle #1,#2 \rangle}}
\newcommand{\oktyp}{\mathrm{ok}}
\newcommand{\size}[1]{\operatorname{size}(#1)}

\newcommand{\stgbracket}[1]{\texttt{.<}#1\texttt{>.}}
\newcommand{\stgescape}[1]{\texttt{.\~{}}#1}
\newcommand{\stgrun}[1]{\texttt{Runcode.run}\: #1}

\newcommand{\eencode}[1]{\ulcorner\hspace{-.5pt}#1\hspace{-.5pt}\urcorner}
\newcommand{\edecode}[1]{\overline{#1}}
\newcommand{\cencode}[1]{\llcorner\hspace{-.5pt}#1\hspace{-.5pt}\lrcorner}
\newcommand{\cdecode}[1]{\underline{#1}}

\newcommand{\func}{\mathit{F}}
\newcommand{\expr}{\mathit{E}}
\newcommand{\cont}{\mathit{C}}
\newcommand{\cst}{\mathit{cst}}
\newcommand{\neutral}{N}

\newcommand{\evar}{\mathit{x}}
\newcommand{\cvar}{a}
\newcommand{\fvar}{\mathit{f}}
\newcommand{\eVar}{\mathit{X}}
\newcommand{\cVar}{\mathit{Y}}

\newcommand{\Fielam}[2]{#1 \Rightarrow #2}
\newcommand{\polop}[2]{\mathord{\mathop{#1}{#2}}}
\newcommand{\plusop}[1]{\polop{+}{#1}}
\newcommand{\minusop}[1]{\polop{-}{#1}}
\newcommand{\NK}{{\textrm{ND}_{\textrm{prop}}}}
\newcommand{\BiNKname}{\textrm{Bi-ND}}
\newcommand{\BiNK}{\BiNKname_{\textrm{prop}}}
\newcommand{\BiNKgets}{\BiNKname^{\mathord{\gets}}_{\textrm{prop}}}
\newcommand{\Rparenthesis}[1]{\mathrm{(}#1\mathrm{)}}
\newcommand{\Noncontradict}{\Rparenthesis{\textrm{Non-contradiction}}}
\newcommand{\Reductio}[1]{\Rparenthesis{\textrm{Reductio}_{#1}}}
\newcommand{\Explosion}[1]{\Rparenthesis{\textrm{Explosion}^{#1}}}
\newcommand{\RAA}[1]{\Rparenthesis{\textrm{RAA}^{#1}}}
\newcommand{\Identity}[1]{\Rparenthesis{\textrm{Identity}_{#1}}}
\newcommand{\RuleLabel}[4]{\Rparenthesis{\mathord{#2}\textrm{-}{\mathrm{#3}_{\mathord{#1}\mathord{#4}}}}}
\newcommand{\Warrow}{\mathrel{\mathord{-}\!\mbox{\raisebox{.8pt}{{\scriptsize $<$}}}}}
\newcommand{\natvdots}{\smash{\vdots}\rule{0pt}{2ex}}

\newcommand{\plusj}[3]{#1 \vdash_{+} #2 \colon #3}
\newcommand{\minusj}[3]{#1 \vdash_{-} #2 \colon #3}
\newcommand{\zeroj}[2]{#1 \vdash_{\hspace{1.5pt}\mathrm{o}\hspace{1.5pt}} #2}

\newcommand{\Ficlam}[2]{#1 \Leftarrow #2}
\newcommand{\Fieapp}[2]{#1 \uparrow #2}
\newcommand{\Ficapp}[2]{#2 \downarrow #1}
\newcommand{\Fiepair}[2]{(#1,#2)}
\newcommand{\Ficpair}[2]{\{#1,#2\}}
\newcommand{\Fieunit}{()}
\newcommand{\Ficnull}{\{\}}
\newcommand{\FieunitT}{\mathrm{unit}}
\newcommand{\FicnullT}{\mathrm{null}}
\newcommand{\uapair}[2]{\mathord{\langle #1 \:|\: #2 \rangle}}
\newcommand{\uatriple}[3]{\mathord{\langle #1 \:|\: #2 \:|\: #3 \rangle}}

\newcommand{\bra}[1]{\mathord{\langle{#1}|}}
\newcommand{\cket}[1]{\mathord{|{#1}\rangle}}
\newcommand{\bracket}[2]{\mathord{\langle #1 \:|\: #2 \rangle}}
\newcommand{\muabs}[2]{\mu #1 . #2}
\newcommand{\sigmaabs}[2]{\varsigma #1 . #2}
\newcommand{\tpvar}{\mathit{tp}}

\newcommand{\bitrans}[1]{(#1)^{\iota}}

\newcommand{\ttemb}{\mathit{h}}
\newcommand{\emb}[2]{\ttemb_{#1}(#2)}
\newcommand{\ttbmi}{\mathit{h}}
\newcommand{\bmi}[2]{\ttbmi_{#1}(#2)}

\newcommand{\remain}{\mathrm{R}}

\newcommand{\leadstogets}{\mathrel{\text{\raisebox{-3pt}{$\overset{\text{\normalsize \rotatebox{180}{$\leadsto$}}}{\leadsto}$}}}}

\newcommand{\BLC}{\textrm{BLC}}
\newcommand{\cbvBLCeq}{\textrm{CbV-BLC}}
\newcommand{\eValue}{V}
\newcommand{\eEvalSymb}{\mathcal{E}}
\newcommand{\cbvBLCred}{\textrm{CbV-BLC}^{\textrm{red}}}
\newcommand{\cbveq}{=_{\mathit{v}}}
\newcommand{\eEval}[1]{\eEvalSymb\{#1\}}

\newcommand{\cbvSubDCname}{\textrm{CbV-DC}}
\newcommand{\cbvSubDCeq}{\cbvSubDCname_{\rightarrow\leftarrow}}
\newcommand{\dcValue}{W}
\newcommand{\dcveq}{=_{\mathit{dcv}}}
\newcommand{\dcEvalSymb}{\mathcal{F}}
\newcommand{\dcEval}[1]{\dcEvalSymb\{#1\}}

\newcommand{\cbvDCred}{\cbvSubDCname^{\textrm{red}}_{\rightarrow\leftarrow}}
\newcommand{\dcvalue}{W^{\mathit{red}}}
\newcommand{\dcvred}{\leadsto}

\newcommand{\dctype}{A}
\newcommand{\dccst}[1]{x_{cst^{#1}}}
\newcommand{\dcbullet}[1]{\alpha_{\bullet^{#1}}}
\newcommand{\dcvar}{x}
\newcommand{\dcimp}{\rightarrow}
\newcommand{\dcpmi}{\leftarrow}
\newcommand{\ot}{\leftarrow}
\newcommand{\dccovar}{\alpha}
\newcommand{\dcterm}{M}
\newcommand{\dccoterm}{K}
\newcommand{\dcstat}{S}
\newcommand{\dcexpr}{O}
\newcommand{\dccovalue}{P}
\newcommand{\dctapp}[2]{#2 \mathbin{\texttt{@}} #1}
\newcommand{\dccapp}[2]{#1 \mathbin{\$} #2}
\newcommand{\dccut}[2]{#1 \mathbin{\bullet} #2}
\newcommand{\dcabs}[2]{#1 \ldotp #2}
\newcommand{\dcfst}[1]{{\tt fst}[#1]}
\newcommand{\dcsnd}[1]{{\tt snd}[#1]}
\newcommand{\dcinl}[1]{\langle#1\rangle{\tt inl}}
\newcommand{\dcinr}[1]{\langle#1\rangle{\tt inr}}
\newcommand{\dcpairr}[2]{\langle#1,#2\rangle}
\newcommand{\dcpairl}[2]{[#1,#2]}
\newcommand{\dcnotr}[1]{[#1]{\tt not}}
\newcommand{\dcnotl}[1]{{\tt not}\langle#1\rangle}
\newcommand{\dcabsr}[2]{(#2).#1}
\newcommand{\dcabsl}[2]{#1.(#2)}
\newcommand{\dcappl}[2]{#1 \mathbin{\texttt{@}} #2}
\newcommand{\dcappr}[2]{#2 \mathbin{\$} #1}

\newcommand{\FV}{{\rm FV}}
\newcommand{\FCV}{{\rm FCV}}
\newcommand{\nyo}[1]{(#1)^{\dagger}}
\newcommand{\fuga}[1]{(#1)^{\sharp}}
\newcommand{\gafu}[1]{(#1)^{\flat}}

\newcommand{\CH}{\bar{\lambda}\mu\tilde{\mu}}

\newcommand{\sterm}{M}
\newcommand{\svalue}{U}
\newcommand{\svar}{x}
\newcommand{\sctxt}{\mathcal{D}}
\newcommand{\stran}[1]{(#1)^{\iota}}

\newcommand{\khterm}{M}
\newcommand{\khvalue}{U}
\newcommand{\khevar}{x}
\newcommand{\khcvar}{x}
\newcommand{\khctxt}{\mathcal{D}}
\newcommand{\khtran}[2]{(#2)^{#1}}

\newcommand{\prompt}[1]{\texttt{prompt}(#1)}
\newcommand{\callcomp}[1]{\texttt{comp}_{#1}}
\newcommand{\callcc}{\texttt{call/cc}}
\newcommand{\control}[1]{\texttt{ctrl}_{#1}}

\newcommand{\ectxt}{\mathcal{E}^{\mathit{red}}}
\newcommand{\cctxt}{\mathcal{C}^{\mathit{red}}}
\newcommand{\ectran}[1]{(#1)^{\Leftarrow}}

\newcommand{\eredValue}{V^{\mathit{red}}}
\newcommand{\vred}{\leadsto}

\newcommand{\dotdiv}{\overset{\cdot}{-}}

\newcommand{\fev}[1]{\operatorname{fev}(#1)}
\newcommand{\fcv}[1]{\operatorname{fcv}(#1)}

\newcommand{\level}[2]{{#2}^{\langle #1 \rangle}}
\newcommand{\levels}[1]{\operatorname{level}(#1)}

\newcommand{\dvar}{\alpha}
\newcommand{\dobj}{t}
\newcommand{\dcom}{T}
\newcommand{\dcst}{c}

\renewcommand{\pair}[2]{\mathord{(#1,#2)}}
\renewcommand{\vec}[1]{\overrightarrow{#1}}
\def\defaultHypSeparation{\hskip 3pt}

\newcommand{\Int}{\texttt{Int}}

\section{Introduction}

A function of the type $A_0 \to A_1$ from expressions of the type
$A_0$ to expressions of the type $A_1$ can be regarded as a function
from continuations of the type $A_1$ to continuations of the type
$A_0$.  This property of functions is called \emph{duality}.

Filinski constructed a symmetric $\lambda$-calculus based on the
duality of
functions~\cite{Filinski:89:DeclarativeContinuations:CTCS,Filinski:mthesis}.
His calculus consists of expressions $\expr$, continuations $\cont$,
and functions $\func$.  Expressions and continuations are
symmetric. Functions are neutral, that is, functions can be encoded to
expressions and continuations like $\eencode{\func}$ and
$\cencode{\func}$, respectively. Expressions and continuations can be
decoded to functions by operators $\edecode{\expr}$ and
$\cdecode{\cont}$.  The operators $\eencode{\cdot}$,
$\cencode{\cdot}$, $\edecode{\cdot}$, and $\cdecode{\cdot}$ are
\emph{primitive} since the duality of functions is adopted as a
\emph{principle} of his calculus.

The duality allows the call-with-current-continuation operator
  (call/cc) to have a type $((A_0 \to A_1) \to A_0) \to A_0$. In a
traditional interpretation of function types, the type means that
call/cc takes an expression of the type $(A_0 \to A_1) \to A_0$ and
returns an expression of the type $A_0$. However, in the symmetric
$\lambda$-calculus, call/cc takes a continuation of the type $A_0$ and
becomes a function of the type $(A_0 \to A_1) \to A_0$, which takes an
expression of the type $A_0 \to A_1$ and returns an expression of the
type $A_0$.

The duality of functions seems to be one of the most significant
reasons that it is possible for the symmetric $\lambda$-calculus to
have the provability of classical logic, because the type
$((A_0 \to A_1) \to A_0) \to A_0$ corresponds to \emph{the Peirce formula}
on the formulae-as-types notion~\cite{curry,HowardW:fortnc}, which
strengthens the $\lambda$-calculus corresponding to the minimal logic
having the provability of classical logic~\cite{Griffin90}.

In this paper, we justify the duality of functions in the symmetric
$\lambda$-calculus using \emph{bilateralism} in proof-theoretic
semantics.  In proof-theoretic semantics there exists an idea that
meanings of logical connectives are given by the contexts in which the
logical connectives occur. In this idea, a meaning of a logical
connective is considered to be defined by its introduction rule of a
natural deduction and its elimination rule is naturally determined to be
\emph{in harmony with} the introduction rule.

Rumfitt suggested that the original natural deduction invented by
Gentzen~\cite{Gentzen1,Gentzen2} is not harmonious, and constructed a
natural deduction based on bilateralism~\cite{rumfitt_i:2000a}. Within
the notion of bilateralism, provability is not defined for a plain
formula $A$ but a formula with polarity $\plusop{A}$ and
$\minusop{A}$.  Provability of $\plusop{A}$ means that $A$ is
accepted, and provability of $\minusop{A}$ means that $A$ is
rejected. The traditional formulation for which provability of $A$
means that $A$ is accepted is based on the notion of
unilateralism rather than bilateralism. Bilateralism does not
permit anything neutral and forces everything to have either positive
or negative polarity.  Rumfitt showed that a natural deduction of
classical logic that is constructed on unilateralism can be
reconstructed on bilateralism.

In this paper, we construct a symmetric $\lambda$-calculus
corresponding to the negation-free bilateral natural deduction.  A
distinguishing aspect of our calculus is that we adopt the but-not
connective as a constructor for functions between continuations.
Another distinguishing aspect is that reductio ad absurdum is a
construction of a configuration also known as a command. In our
calculus, continuations and commands are first-class citizens.

Our bilateral $\lambda$-calculus contains a computationally consistent
call-by-value calculus.
The calculus corresponds to the sub-calculus of the call-by-value dual
calculus invented by Wadler~\cite{journals/sigplan/Wadler03,Wadler05}
obtained by adding the but-not connective and removing
the negation connective.
The equivalence is formally obtained by giving mutual translations
between these calculi.  In other words, the translation provides a
strong relationship between a bilateral natural deduction and a
sequent calculus including proofs on the formulae-as-types
notion.

The translations clarify a significant difference between the
bilateral natural deduction and the sequent calculus.
The negation of the dual calculus is not involutive, that is,
$\negop{\negop{A}}$ is not isomorphic to $A$.
Although the dual calculus also has the involutive duality as the
meta-level operation that comes from the left-hand-side and
right-hand-side duality of the classical sequent-calculus framework,
there exists no inference rule to operate the involutive duality in
the calculus.
In the bilateral $\lambda$-calculus, the negation is represented using
inversions of polarities, and is involutive by definition.

A symmetric $\lambda$-calculus which was constructed by Lovas and
Crary is the only similar calculus based on
bilateralism~\cite{lovas2006}. However, they adopted the negation
connective $\neg$ as a primitive logical connective, and function type
$\to$ is defined as syntactic sugar.  In Lovas and Crary's calculus it is
necessary to use reductio ad absurdum, although it is generally easy
to define functions between expressions. This means that it is not
easy to define a sub-calculus corresponding to the minimal logic.  Our
calculus does not include the negation connective. Our work claims
that the negation connective is not necessary but negative polarity is
sufficient to define a symmetric $\lambda$-calculus based on
bilateralism.




Using our calculus, we justify the duality which Filinski adopted as a
principle in constructing his calculus. Specifically, the encodings to
expressions and continuations are \emph{definable} in our
calculus. More correctly, mutual transformations between expressions
and continuations of function types are definable in our calculus.  We
also show that every typable function has dual types about expressions
and continuations.  We clarify that bilateralism naturally raises the
duality of functions.



Finally, we note that one of our goals is to construct a simple and
powerful calculus in which the duality of functions is definable. We
do not intend to clarify anything unknown in classical logic by
assigning $\lambda$-terms to proofs, as seen in existing work in
structural proof theory.  Actually, our calculus is a sub-calculus of
a natural extension of the dual calculus.

The remainder of this paper is organized as follows: In
Section~\ref{sec:bilateralism}, we introduce bilateral natural
deductions. In Section~\ref{sec:curry}, we add proofs to nodes in
derivation trees. In Section~\ref{sec:blc}, we construct a symmetric
$\lambda$-calculus corresponding to the negation-free bilateral natural
deduction.  In Section~\ref{sec:neutral}, we justify the duality of
functions using our calculus.  In Section~\ref{sec:related}, we
discuss related work to clarify the contributions of this paper.  In
Section~\ref{sec:conclusion}, we conclude the paper by identifying
future research directions.

\section{Bilateral Natural Deductions}\label{sec:bilateralism}

In this section, we introduce \emph{bilateralism}, which was proposed
by Rumfitt~\cite{rumfitt_i:2000a}, and define a few variants of Rumfitt's bilateral
natural deduction.

The set of formulae is defined as follows:
\begin{alignat*}{2}
  & \mbox{(formulae)} & \quad A & \Coloneqq o \mid (\negop{A}) \mid (A
  \to A) \mid (A \wedge A) \mid (A \vee A)
\end{alignat*}
where $o$ ranges over propositional variables. We note that $\bot$ is not contained
by the set of formulae.  The connective power of $\neg$ is stronger
than that of $\wedge$, $\vee$, and $\to$.  The connective powers of
$\wedge$ and $\vee$ are stronger than that of $\to$.
We omit parentheses when the context renders them obvious.

\begin{figure}[t]
  \begin{center}
    \AxiomC{$A$}
    \AxiomC{$\negop{A}$}
    \RightLabel{$\RuleLabel{}{\bot}{I}{}$}
    \BinaryInfC{$\bot$}
    \bottomAlignProof
    \DisplayProof
    \qquad
    \AxiomC{$\bot$}
    \RightLabel{$\RuleLabel{}{\bot}{E}{}$}
    \UnaryInfC{$A$}
    \bottomAlignProof
    \DisplayProof
    \qquad
    \AxiomC{$[A]$}
    \noLine
    \UnaryInfC{$\natvdots$}
    \noLine
    \UnaryInfC{$\bot$}
    \RightLabel{$\RuleLabel{}{\neg}{I}{}$}
    \UnaryInfC{$\negop{A}$}
    \bottomAlignProof
    \DisplayProof
    \qquad
    \AxiomC{$[\negop{A}]$}
    \noLine
    \UnaryInfC{$\natvdots$}
    \noLine
    \UnaryInfC{$\bot$}
    \RightLabel{$\RuleLabel{}{\neg}{E}{}$}
    \UnaryInfC{$A$}
    \bottomAlignProof
    \DisplayProof
    \\\vspace{\baselineskip}
    \AxiomC{$[A_0]$}
    \noLine
    \UnaryInfC{$\natvdots$}
    \noLine
    \UnaryInfC{$A_1$}
    \RightLabel{$\RuleLabel{}{\to}{I}{}$}
    \UnaryInfC{$A_0 \to A_1$}
    \bottomAlignProof
    \DisplayProof
    \;
    \AxiomC{$A_0 \to A_1$}
    \AxiomC{$A_0$}
    \RightLabel{$\RuleLabel{}{\to}{E}{}$}
    \BinaryInfC{$A_1$}
    \bottomAlignProof
    \DisplayProof
    \;
    \AxiomC{$A_0$}
    \AxiomC{$A_1$}
    \RightLabel{$\RuleLabel{}{\wedge}{I}{}$}
    \BinaryInfC{$A_0 \wedge A_1$}
    \bottomAlignProof
    \DisplayProof
    \;
    \AxiomC{$A_0 \wedge A_1$}
    \RightLabel{$\RuleLabel{}{\wedge}{E}{0}$}
    \UnaryInfC{$A_0$}
    \bottomAlignProof
    \DisplayProof
    \\\vspace{\baselineskip}
    \AxiomC{$A_0 \wedge A_1$}
    \RightLabel{$\RuleLabel{}{\wedge}{E}{1}$}
    \UnaryInfC{$A_1$}
    \bottomAlignProof
    \DisplayProof
    \;
    \AxiomC{$A_0$}
    \RightLabel{$\RuleLabel{}{\vee}{I}{0}$}
    \UnaryInfC{$A_0 \vee A_1$}
    \bottomAlignProof
    \DisplayProof
    \;\;
    \AxiomC{$A_1$}
    \RightLabel{$\RuleLabel{}{\vee}{I}{1}$}
    \UnaryInfC{$A_0 \vee A_1$}
    \bottomAlignProof
    \DisplayProof
    \;\;
    \AxiomC{$A_0 \vee A_1$}
    \AxiomC{$[A_0]$}
    \noLine
    \UnaryInfC{$\natvdots$}
    \noLine
    \UnaryInfC{$A_2$}
    \AxiomC{$[A_1]$}
    \noLine
    \UnaryInfC{$\natvdots$}
    \noLine
    \UnaryInfC{$A_2$}
    \RightLabel{$\RuleLabel{}{\vee}{E}{}$}
    \TrinaryInfC{$A_2$}
    \bottomAlignProof
    \DisplayProof
  \end{center}
  \caption{Natural deduction $\NK$.}\label{fig:NK}
\end{figure}

We recall the natural deduction invented by
Gentzen~\cite{Gentzen1,Gentzen2} and consider its propositional
fragment $\NK$, as shown in Figure~\ref{fig:NK}. At each inference
rule, formulae or $\bot$ above a line are assumptions and a formula or
$\bot$ below a line is a conclusion.  A derivation is a tree that has
exactly one root. Symbol $\smash{\vdots}$ denotes a transitive
connection between a leaf and a node, and $[A]$ means that $A$ is
discharged from assumptions in a standard manner.  Rules
$\RuleLabel{}{\bot}{E}{}$ and $\RuleLabel{}{\neg}{E}{}$ are also known
as \emph{explosion} and \emph{reductio ad absurdum}, respectively.  A
judgment is defined as $\varGamma \vdash A$ or $\varGamma \vdash
\bot$, where $\varGamma$ is a multiset of formulae.

There exists an idea that meanings of logical connectives are defined
by their introduction rules and their elimination rules should be
defined \emph{in harmony with} their introduction rules in
proof-theoretic semantics.
Rumfitt attempted to justify logical connectives and inference rules
using a notion of harmony which was proposed by
Dummett~\cite{Dum:logbm}.  We consider a logical connective
$\mathit{tonk}$ which was proposed by Prior~\cite{prior60analysis}.
Its introduction rule $\RuleLabel{}{\mathit{tonk}}{I}{}$ and
elimination rule $\RuleLabel{}{\mathit{tonk}}{E}{}$ are as follows:
\begin{center}
    \AxiomC{$A_0$}
    \RightLabel{$\RuleLabel{}{\mathit{tonk}}{I}{}$}
    \UnaryInfC{$A_0 \mathbin{\mathit{tonk}} A_1$}
    \DisplayProof
    \qquad
    \AxiomC{$A_0 \mathbin{\mathit{tonk}} A_1$}
    \RightLabel{$\RuleLabel{}{\mathit{tonk}}{E}{}$}
    \UnaryInfC{$A_1$}
    \DisplayProof
    \enspace .
\end{center}

A pair of contiguous introduction and elimination rules is called
\emph{harmonious} if the residue after removing the pair is also a
derivation. Such a procedure is called \emph{normalization}. In this
section, we let $\vred$ denote the normalization procedure.  The
pair of $\RuleLabel{}{\mathit{tonk}}{I}{}$ and
$\RuleLabel{}{\mathit{tonk}}{E}{}$ is not harmonious because the
right-hand side of the following $\vred$ relation is not a derivation:
\begin{center}
    \AxiomC{$A_0$}
    \RightLabel{$\RuleLabel{}{\mathit{tonk}}{I}{}$}
    \UnaryInfC{$A_0 \mathbin{\mathit{tonk}} A_1$}
    \RightLabel{$\RuleLabel{}{\mathit{tonk}}{E}{}$}
    \UnaryInfC{$A_1$}
    \DisplayProof
    $\;\; \vred$
    \AxiomC{$A_0$}
    \UnaryInfC{$A_1$}
    \DisplayProof
    \enspace .
\end{center}

Rumfitt suggested that $\NK$ also does not enjoy the harmony
condition and proposed a notion of \emph{bilateralism} to construct a
harmonious natural deduction.

Bilateralism is based on two notions of \emph{acceptance} and
\emph{rejection} of formulae.  They are also called
\emph{verification} and \emph{falsification}, respectively,
by Wansing~\cite{Wansingaiml,journals/logcom/Wansing16}. Formulae $A$ with
\emph{polarity} are defined as $\plusop{A}$ and $\minusop{A}$.  A
derivation of root $\plusop{A}$ means that $A$ is accepted. A
derivation of root $\minusop{A}$ means that $A$ is rejected.

Let $\mathcal{A}$ be a formula with polarity. Conjugates
$(\plusop{A})^\ast$ and $(\minusop{A})^\ast$ are defined as
$\minusop{A}$ and $\plusop{A}$, respectively.

Rumfitt adopted $\Noncontradict$ and $\Reductio{}$ which are called
\emph{coordination principles}
and defined inference rules of logical connectives, as shown in
Figure~\ref{fig:BiNK}, which are naturally derived from the standard
boolean semantics. In this paper, we call this logic a bilateral
natural deduction $\BiNK$.

$\NK$ is based on the notion of \emph{unilateralism} rather than
bilateralism. A derivation of root $A$ in $\NK$ means that $A$ is
accepted. There exists the following relation between $\NK$ and $\BiNK$:
\begin{theorem}[Rumfitt~\cite{rumfitt_i:2000a}]
  For any $n \geq 0$, $A_0,\ldots,A_{n-1} \vdash A$ is provable in $\NK$ if and only if
  $\plusop{A_0},\ldots,\plusop{A_{n-1}} \vdash \plusop{A}$ is provable in
  $\BiNK$.
\end{theorem}

\noindent
\emph{Remark.} It is controversial that explosion and reductio ad
absurdum are regarded as elimination rules of the logical connectives
$\bot$ and $\neg$, respectively.  Rumfitt's bilateralism is also
criticized in a paper~\cite{journals/jphil/Kurbis16}. That is,
bilateralism is called a work in progress. However, the subject of
this paper is not a justification of bilateralism in proof-theoretic
semantics.

\begin{figure}[t]
\begin{center}
    \AxiomC{$\phantom{a}$}
    \noLine
    \UnaryInfC{$\phantom{a}$}
    \noLine
    \UnaryInfC{$\phantom{a}$}
    \noLine
    \UnaryInfC{$\mathcal{A}$}
    \AxiomC{$\mathcal{A}^\ast$}
    \RightLabel{$\Noncontradict$}
    \BinaryInfC{$\bot$}
    \DisplayProof
    \qquad\qquad\qquad
    \AxiomC{$[\mathcal{A}]$}
    \noLine
    \UnaryInfC{$\natvdots$}
    \noLine
    \UnaryInfC{$\bot$}
    \RightLabel{$\Reductio{}$}
    \UnaryInfC{$\mathcal{A}^\ast$}
    \DisplayProof
\end{center}
\begin{minipage}{.49\textwidth}
\begin{center}
    \AxiomC{$\minusop{A}$}
    \RightLabel{$\RuleLabel{+}{\neg}{I}{}$}
    \UnaryInfC{$\plusop{\negop{A}}$}
    \bottomAlignProof
    \DisplayProof
    \qquad
    \AxiomC{$\plusop{\negop{A}}$}
    \RightLabel{$\RuleLabel{+}{\neg}{E}{}$}
    \UnaryInfC{$\minusop{A}$}
    \bottomAlignProof
    \DisplayProof
    \\\vspace{\baselineskip}
    \AxiomC{$[\plusop{A_0}]$}
    \noLine
    \UnaryInfC{$\natvdots$}
    \noLine
    \UnaryInfC{$\plusop{A_1}$}
    \RightLabel{$\RuleLabel{+}{\to}{I}{}$}
    \UnaryInfC{$\plusop{A_0 \to A_1}$}
    \bottomAlignProof
    \DisplayProof
    \\\vspace{\baselineskip}
    \AxiomC{$\plusop{A_0 \to A_1}$}
    \AxiomC{$\plusop{A_0}$}
    \RightLabel{$\RuleLabel{+}{\to}{E}{}$}
    \BinaryInfC{$\plusop{A_1}$}
    \bottomAlignProof
    \DisplayProof
    \\\vspace{\baselineskip}
    \AxiomC{$\plusop{A_0}$}
    \AxiomC{$\plusop{A_1}$}
    \RightLabel{$\RuleLabel{+}{\wedge}{I}{}$}
    \BinaryInfC{$\plusop{A_0 \wedge A_1}$}
    \bottomAlignProof
    \DisplayProof
    \\\vspace{\baselineskip}
    \AxiomC{$\plusop{A_0 \wedge A_1}$}
    \RightLabel{$\RuleLabel{+}{\wedge}{E}{0}$}
    \UnaryInfC{$\plusop{A_0}$}
    \bottomAlignProof
    \DisplayProof
    \;
    \AxiomC{$\plusop{A_0 \wedge A_1}$}
    \RightLabel{$\RuleLabel{+}{\wedge}{E}{1}$}
    \UnaryInfC{$\plusop{A_1}$}
    \bottomAlignProof
    \DisplayProof
    \\\vspace{\baselineskip}
    \AxiomC{$\plusop{A_0}$}
    \RightLabel{$\RuleLabel{+}{\vee}{I}{0}$}
    \UnaryInfC{$\plusop{A_0 \vee A_1}$}
    \bottomAlignProof
    \DisplayProof
    \quad
    \AxiomC{$\plusop{A_1}$}
    \RightLabel{$\RuleLabel{+}{\vee}{I}{1}$}
    \UnaryInfC{$\plusop{A_0 \vee A_1}$}
    \bottomAlignProof
    \DisplayProof
    \\\vspace{\baselineskip}
    \AxiomC{$\plusop{A_0 \vee A_1}$}
    \AxiomC{$[\plusop{A_0}]$}
    \noLine
    \UnaryInfC{$\natvdots$}
    \noLine
    \UnaryInfC{$\mathcal{A}$}
    \AxiomC{$[\plusop{A_1}]$}
    \noLine
    \UnaryInfC{$\natvdots$}
    \noLine
    \UnaryInfC{$\mathcal{A}$}
    \RightLabel{$\RuleLabel{+}{\vee}{E}{}$}
    \TrinaryInfC{$\mathcal{A}$}
    \bottomAlignProof
    \DisplayProof
  \end{center}
\end{minipage}
\begin{minipage}{.49\textwidth}
\begin{center}
    \AxiomC{$\plusop{A}$}
    \RightLabel{$\RuleLabel{-}{\neg}{I}{}$}
    \UnaryInfC{$\minusop{\negop{A}}$}
    \bottomAlignProof
    \DisplayProof
    \qquad
    \AxiomC{$\minusop{\negop{A}}$}
    \RightLabel{$\RuleLabel{-}{\neg}{E}{}$}
    \UnaryInfC{$\plusop{A}$}
    \bottomAlignProof
    \DisplayProof
    \\\vspace{\baselineskip}
    \AxiomC{$\plusop{A_0}$}
    \AxiomC{$\minusop{A_1}$}
    \RightLabel{$\RuleLabel{-}{\to}{I}{}$}
    \BinaryInfC{$\minusop{A_0 \to A_1}$}
    \bottomAlignProof
    \DisplayProof
    \\\vspace{\baselineskip}
    \AxiomC{$\minusop{A_0 \to A_1}$}
    \RightLabel{$\RuleLabel{-}{\to}{E}{0}$}
    \UnaryInfC{$\plusop{A_0}$}
    \bottomAlignProof
    \DisplayProof
    \\\vspace{\baselineskip}
    \AxiomC{$\minusop{A_0 \to A_1}$}
    \RightLabel{$\RuleLabel{-}{\to}{E}{1}$}
    \UnaryInfC{$\minusop{A_1}$}
    \bottomAlignProof
    \DisplayProof
    \\\vspace{\baselineskip}
    \AxiomC{$\minusop{A_0}$}
    \RightLabel{$\RuleLabel{-}{\wedge}{I}{0}$}
    \UnaryInfC{$\minusop{A_0 \wedge A_1}$}
    \bottomAlignProof
    \DisplayProof
    \quad
    \AxiomC{$\minusop{A_1}$}
    \RightLabel{$\RuleLabel{-}{\wedge}{I}{1}$}
    \UnaryInfC{$\minusop{A_1 \wedge A_1}$}
    \bottomAlignProof
    \DisplayProof
    \\\vspace{\baselineskip}
    \AxiomC{$\minusop{A_0 \wedge A_1}$}
    \AxiomC{$[\minusop{A_0}]$}
    \noLine
    \UnaryInfC{$\natvdots$}
    \noLine
    \UnaryInfC{$\mathcal{A}$}
    \AxiomC{$[\minusop{A_1}]$}
    \noLine
    \UnaryInfC{$\natvdots$}
    \noLine
    \UnaryInfC{$\mathcal{A}$}
    \RightLabel{$\RuleLabel{-}{\wedge}{E}{}$}
    \TrinaryInfC{$\mathcal{A}$}
    \bottomAlignProof
    \DisplayProof
    \\\vspace{\baselineskip}
    \AxiomC{$\minusop{A_0}$}
    \AxiomC{$\minusop{A_1}$}
    \RightLabel{$\RuleLabel{-}{\vee}{I}{}$}
    \BinaryInfC{$\minusop{A_0 \vee A_1}$}
    \bottomAlignProof
    \DisplayProof
    \\\vspace{\baselineskip}
    \AxiomC{$\minusop{A_0 \vee A_1}$}
    \RightLabel{$\RuleLabel{-}{\vee}{E}{0}$}
    \UnaryInfC{$\minusop{A_0}$}
    \bottomAlignProof
    \DisplayProof
    \;
    \AxiomC{$\minusop{A_0 \vee A_1}$}
    \RightLabel{$\RuleLabel{-}{\vee}{E}{1}$}
    \UnaryInfC{$\minusop{A_1}$}
    \bottomAlignProof
    \DisplayProof
\end{center}
\end{minipage}
  \caption{Rumfitt's natural deduction $\BiNK$.}\label{fig:BiNK}
\end{figure}

The natural deduction $\BiNK$ is not symmetric.
We extend the language by adding a logical connective $\gets$:
\begin{alignat*}{2}
  & \mbox{(formulae)} & \quad A & \Coloneqq o \mid (\negop{A}) \mid (A \to A) \mid (A \gets A) \mid (A \wedge A) \mid (A \vee A) \enspace .
\end{alignat*}
We will use the logical connective as function types of continuations
in the following.

The connective $\gets$ is called the \emph{but-not} connective because
$A_0 \gets A_1$ is logically equivalent to $A_0 \wedge \negop{A_1}$ in
classical logic. The but-not connective is also written as
pseudo-difference $\dotdiv$~\cite{Goodman,Urbas}, subtraction
$-$~\cite{oai:CiteSeerPSU:230037}, difference $-$~\cite{CurienICFP00},
and co-implication
$\Warrow$~\cite{conf/aiml/GorePT08,journals/logcom/Wansing16}.
The but-not connective is a primitive connective in paraconsistent
logic, whereas $\to$ is a primitive connective in intuitionistic logic
because $A_0 \to A_1$ is not logically equivalent to $\negop{A_0} \vee
A_1$ in intuitionistic logic.  In paraconsistent logic, sequent
calculus consists of sequents $\varGamma \vdash \varDelta$, where
$\varGamma$ is empty or a singleton formula, whereas intuitionistic
logic can be defined by sequents $\varGamma \vdash \varDelta$, where
$\varDelta$ is empty or a singleton formula.

\begin{figure}[t]
  \begin{center}
    \AxiomC{$\plusop{A_0}$}
    \AxiomC{$\minusop{A_1}$}
    \RightLabel{$\RuleLabel{+}{\gets}{I}{}$}
    \BinaryInfC{$\plusop{A_0 \gets A_1}$}
    \bottomAlignProof
    \DisplayProof
    \qquad
    \AxiomC{$\plusop{A_0 \gets A_1}$}
    \RightLabel{$\RuleLabel{+}{\gets}{E}{0}$}
    \UnaryInfC{$\plusop{A_0}$}
    \bottomAlignProof
    \DisplayProof
    \qquad
    \AxiomC{$\plusop{A_0 \gets A_1}$}
    \RightLabel{$\RuleLabel{+}{\gets}{E}{1}$}
    \UnaryInfC{$\minusop{A_1}$}
    \bottomAlignProof
    \DisplayProof
    \\\vspace{\baselineskip}
    \AxiomC{$[\minusop{A_1}]$}
    \noLine
    \UnaryInfC{$\natvdots$}
    \noLine
    \UnaryInfC{$\minusop{A_0}$}
    \RightLabel{$\RuleLabel{-}{\gets}{I}{}$}
    \UnaryInfC{$\minusop{A_0 \gets A_1}$}
    \bottomAlignProof
    \DisplayProof
    \qquad
    \AxiomC{$\minusop{A_0 \gets A_1}$}
    \AxiomC{$\minusop{A_1}$}
    \RightLabel{$\RuleLabel{-}{\gets}{E}{}$}
    \BinaryInfC{$\minusop{A_0}$}
    \bottomAlignProof
    \DisplayProof
  \end{center}
\caption{Inference rules for $\gets$.}\label{fig:gets}
\end{figure}

We define a natural deduction $\BiNKgets$ by adding inference rules, as
shown in Figure~\ref{fig:gets}.
The connectives $\to$ and $\gets$ are symmetrically located in
$\BiNKgets$ as follows:
\begin{proposition}\label{prop:plusminus}
  $\plusop{A_0 \to A_1} \vdash \minusop{A_0 \gets A_1}$,
  $\minusop{A_0 \to A_1} \vdash \plusop{A_0 \gets A_1}$,
  $\plusop{A_0 \gets A_1} \vdash \minusop{A_0 \to A_1}$, and
  $\minusop{A_0 \gets A_1} \vdash \plusop{A_0 \to A_1}$
  are provable in $\BiNKgets$.
\end{proposition}
\begin{proof}
See Appendix~\ref{sec:proofs}.
\end{proof}

Let $\mathfrak{L}_0$ and $\mathfrak{L}_1$ be languages such that
$\mathfrak{L}_0 \subseteq \mathfrak{L}_1$, and $\mathfrak{S}_0$ and
$\mathfrak{S}_1$ be logics on the languages $\mathfrak{L}_0$ and
$\mathfrak{L}_1$, respectively.  We define $\mathfrak{S}_1$ as an
extension of $\mathfrak{S}_0$ if any formula $\varphi$ that is
provable in $\mathfrak{S}_0$ is also provable in $\mathfrak{S}_1$. We
define that an extension $\mathfrak{S}_1$ of $\mathfrak{S}_0$ is
\emph{conservative} if any formula $\varphi$ on the language $\mathfrak{L}_0$
that is provable in $\mathfrak{S}_1$ is also provable in
$\mathfrak{S}_0$.
\begin{proposition}
  $\BiNKgets$ is a conservative extension of $\BiNK$.
\end{proposition}
\begin{proof}
  It is obvious because $\BiNK$ is complete to the standard two-value
  semantics, and $\BiNKgets$ is sound to the semantics.
\end{proof}

$\BiNK$ and $\BiNKgets$ include sub-logics as follows:

\begin{proposition}\label{prop:redundant}
  \begin{enumerate}
  \item
  The inference rules
  $\RuleLabel{-}{\wedge}{I}{0}$,
  $\RuleLabel{-}{\wedge}{I}{1}$,
  $\RuleLabel{-}{\wedge}{E}{}$,
  $\RuleLabel{+}{\vee}{I}{}$,
  $\RuleLabel{+}{\vee}{E}{0}$, and
  $\RuleLabel{+}{\vee}{E}{1}$
  are derivable in $\BiNK$, and
  \item
  The inference rules
  $\RuleLabel{-}{\to}{I}{}$,
  $\RuleLabel{-}{\to}{E}{0}$,
  $\RuleLabel{-}{\to}{E}{1}$,
  $\RuleLabel{+}{\gets}{I}{}$,
  $\RuleLabel{+}{\gets}{E}{0}$, and
  $\RuleLabel{+}{\gets}{E}{1}$
  are derivable in $\BiNKgets$.
  \end{enumerate}
  \end{proposition}
\begin{proof}
See Appendix~\ref{sec:proofs}.
\end{proof}

\section{Derivation Trees with Proofs in Their Nodes}\label{sec:curry}

In this section, we introduce derivation trees with proofs in their
nodes to mediate between natural deductions and $\lambda$-calculi
introduced in Sections~\ref{sec:bilateralism} and \ref{sec:blc},
respectively.

We add proofs to polarized formulae in the $\neg$-free fragment of
$\BiNKgets$, that is, the $\Noncontradict$, $\Reductio{}$,
$\RuleLabel{+}{\to}{I}{}$, $\RuleLabel{+}{\to}{E}{}$,
$\RuleLabel{+}{\wedge}{I}{}$, $\RuleLabel{+}{\wedge}{E}{0}$,
$\RuleLabel{+}{\wedge}{E}{1}$, $\RuleLabel{-}{\vee}{I}{}$,
$\RuleLabel{-}{\vee}{E}{0}$, $\RuleLabel{-}{\vee}{E}{1}$,
$\RuleLabel{+}{\gets}{I}{}$, $\RuleLabel{+}{\gets}{E}{0}$,
$\RuleLabel{+}{\gets}{E}{1}$, $\RuleLabel{-}{\gets}{I}{}$, and
$\RuleLabel{-}{\gets}{E}{}$ fragment, to construct a symmetric
$\lambda$-calculus.  We note that the other inference rules are
derivable by Proposition~\ref{prop:redundant}.

We assume a set of \emph{proof variables}. We write $\dvar$ for a
proof variable. We define that nodes $\dobj \colon \plusop{A}$ and
$\dobj \colon \minusop{A}$ in the natural deduction respectively
denote that $\dobj$ is a proof for acceptance and rejection of $A$. We
also define that a node $\dcom \colon \bot$ in the natural deduction
denotes that $\dcom$ is a proof for contradiction.

Node $\lamabs{\dvar}{\dobj} \colon \plusop{A_0 \to A_1}$ denotes that
$\lamabs{\dvar}{\dobj}$ is a proof for acceptance of $A_0 \to A_1$ if
$\dvar$ is a proof variable for acceptance of $A_0$ and $\dobj$ is a
proof for acceptance of $A_1$.
Node $\lamabs{\dvar}{\dobj} \colon \minusop{A_0 \gets A_1}$ denotes
that $\lamabs{\dvar}{\dobj}$ is a proof for acceptance of $A_0 \to
A_1$ if $\dvar$ is a proof variable for acceptance of $A_1$ and $\dobj$ is a
proof for acceptance of $A_0$.

Node $\app{\dobj_0}{\dobj_1} \colon \plusop{A_1}$ denotes that
$\app{\dobj_0}{\dobj_1}$ is a proof for acceptance of $A_1$ if
$\dobj_0$ is a proof for acceptance of $A_0 \to A_1$ and $\dobj_1$ is a proof
for acceptance of $A_0$.
Node $\app{\dobj_0}{\dobj_1} \colon \minusop{A_0}$ denotes that
$\app{\dobj_0}{\dobj_1}$ is a proof for rejection of $A_1$ if
$\dobj_0$ is a proof for rejection of $A_0 \gets A_1$ and $\dobj_1$ is a
proof for rejection of $A_1$.

Node $\pair{\dobj_0}{\dobj_1} \colon \plusop{A_0 \wedge A_1}$ denotes
that $\pair{\dobj_0}{\dobj_1}$ is a proof for acceptance of $A_0 \wedge
A_1$ if $\dobj_0$ is a proof for acceptance of $A_0$ and $\dobj_1$ is
a proof for acceptance of $A_1$.
Node $\pair{\dobj_0}{\dobj_1} \colon \minusop{A_0 \vee A_1}$ denotes
that $\pair{\dobj_0}{\dobj_1}$ is a proof for rejection of $A_0 \vee
A_1$ if $\dobj_0$ is a proof for rejection of $A_0$ and $\dobj_1$ is a
proof for rejection of $A_1$.

Node $\fst{\dobj} \colon \plusop{A_0}$ denotes that $\fst{\dobj}$ is a
proof for acceptance of $A_0$ if $\dobj_0$ is a proof for acceptance
of $A_0 \wedge A_1$.
Node $\fst{\dobj_0} \colon \minusop{A_0}$ denotes that $\fst{\dobj_0}$
is a proof for rejection of $A_0$ if $\dobj_0$ is a proof for
rejection of $A_0 \vee A_1$.
Nodes $\snd{\dobj_0} \colon \plusop{A_1}$ and $\snd{\dobj_0} \colon
\minusop{A_1}$ are similar.

Node $\bracket{\dobj_0}{\dobj_1} \colon \bot$ denotes that
$\bracket{\dobj_0}{\dobj_1}$ is a proof of contradiction if $\dobj_0$ is a
proof for acceptance of $A$ and $\dobj_1$ is a proof for rejection of
$A_0$.

Node $\muabs{\dvar}{\dcom} \colon \plusop{A}$ denotes that
$\muabs{\dvar}{\dcom}$ is a proof for acceptance of $A$ if $\dvar$ is
a proof variable for rejection of $A$ and $\dcom$ is a proof of
contradiction.
Node $\muabs{\dvar}{\dcom} \colon \minusop{A}$ denotes that
$\muabs{\dvar}{\dcom}$ is a proof for rejection of $A$ if $\dvar$ is a
proof variable for acceptance of $A$ and $\dcom$ is a proof of
contradiction.

We formally define the set of proofs in a Curry-style bilateral
$\lambda$-calculus, as shown in Figure~\ref{fig:proofs} where $\dcst$
ranges over constants for adding logical axioms.

\begin{figure}
\begin{align*}
  & \mbox{(proofs)} &
  \dobj
  & \Coloneqq \dcst
  \mid \dvar
  \mid \lamabs{\dvar}{\dobj}
  \mid \app{\dobj}{\dobj}
  \mid \pair{\dobj}{\dobj} \mid \fst{\dobj}
  \mid \snd{\dobj}
  \mid \muabs{\dvar}{\dcom} &
  \dcom
  & \Coloneqq \bracket{\dobj}{\dobj}
\end{align*}
\caption{Proofs of the negation-free natural
  deduction.}\label{fig:proofs}
\end{figure}

Let $\varGamma$ be a set of nodes.  Judgment $\varGamma \vdash \dobj
\colon \plusop{A}$ denotes that $\dobj$ is a proof for acceptance of
$A$ under $\varGamma$. Judgment $\varGamma \vdash \dobj \colon
\minusop{A}$ denotes that $\dobj$ is a proof for rejection of $A$
under $\varGamma$. Judgment $\varGamma \vdash \dcom \colon \bot$
denotes that $\dcom$ is a proof for contradiction under $\varGamma$.

\section{Bilateral Lambda-Calculi}\label{sec:blc}

In this section, we construct a Church-style symmetric
$\lambda$-calculus based on bilateralism and define a call-by-value
sub-calculus.

\subsection{Definition and Basic Properties}

We respectively call proofs for acceptance and rejection
\emph{expressions} and \emph{continuations}.  We distinguish proof
variables for acceptance from those for rejection.  We construct an
alternative symmetric $\lambda$-calculus called a bilateral
$\lambda$-calculus ($\BLC$).

We define types, polarized types, expressions, continuations,
commands, and syntactical objects as shown in Figure~\ref{fig:blc}.


\begin{figure}[t]
\begin{align*}
  & \mbox{(types)} & A & \Coloneqq o \mid (A \to A) \mid (A \gets A) \mid (A \wedge A) \mid (A \vee A)\\
  & \mbox{(expressions)} &
  \expr
  & \Coloneqq \cst^o
  \mid \evar^A
  \mid \lamabs{\evar^A}{\expr}
  \mid \app{\expr}{\expr}
  \mid \pair{\expr}{\expr} \mid \fst{\expr}
  \mid \snd{\expr}
  \mid \muabs{\cvar^A}{\neutral}\\
  & \mbox{(continuations)} &
  \cont
  & \Coloneqq \bullet^o
  \mid \cvar^A
  \mid \lamabs{\cvar^A}{\cont}
  \mid \app{\cont}{\cont}
  \mid \pair{\cont}{\cont} \mid \fst{\cont}
  \mid \snd{\cont}
  \mid \muabs{\evar^A}{\neutral}\\
  & \mbox{(commands)} &
  \neutral
  & \Coloneqq \bracket{\expr}{\cont}\\
  & \mbox{(syntactical objects)} &
  D
  & \Coloneqq \expr \mid \cont \mid \neutral
\end{align*}
\caption{The bilateral lambda-calculus $\BLC$.}\label{fig:blc}
\end{figure}

Expression $\cst^o$ denotes a constant. Expression $\evar^A$ denotes
an expression variable.
Expression $\lamabs{\evar^{A}}{\expr}$ denotes a
$\lambda$-abstraction of expression $\expr$ by
$\evar^{A}$.
Expression $\app{\expr_0}{\expr_1}$ denotes an application of function
$\expr_0$ to expression $\expr_1$.
Expression
$\pair{\expr_0}{\expr_1}$ denotes a pair of expressions $\expr_0$ and
$\expr_1$.
Expressions $\fst{\expr}$ and $\snd{\expr}$ are projections.

Continuations are defined symmetrically to expressions. Continuation
$\bullet^o$ denotes the unique constant denoting a continuation of
$o$. By the definition based on bilateralism, the calculus is
involutive on the notion of polarities.

Commands are first-class citizens. A command can be abstracted by
expression variable $\evar^A$ or continuation variable $\cvar^A$.  Command
$\neutral$ abstracted by $\cvar^A$ is expression
$\muabs{\cvar^A}{\neutral}$. A command abstracted by $\evar^A$ is
continuation $\muabs{\evar^A}{\neutral}$.  A similar idea can be seen in
$\CH$-calculus which was proposed by Curien and Herbelin~\cite{CurienICFP00}.


Expressions, continuations, and commands are called syntactical objects.

We assume that the connective powers of applications are stronger than
those of $\lambda$-abstractions.  We omit superscripts that
denote types when the context renders them obvious.

\begin{figure*}[t]
  \begin{center}
    \AxiomC{$\plusj{\varGamma}{\expr}{A}$}
    \AxiomC{$\minusj{\varGamma}{\cont}{A}$}
    \RightLabel{$\Noncontradict$}
    \BinaryInfC{$\zeroj{\varGamma}{\bracket{\expr}{\cont}}$}
    \DisplayProof
  \end{center}
  \begin{minipage}{.49\textwidth}
  \begin{center}
    \AxiomC{$\zeroj{\varPi; \varSigma, \cvar \colon A}{\neutral}$}
    \RightLabel{$\Reductio{+}$}
    \UnaryInfC{$\plusj{\varPi; \varSigma}{\muabs{\cvar^A}{\neutral}}{A}$}
    \DisplayProof
    \\\vspace{\baselineskip}
    \AxiomC{\phantom{$X^{Y^Z}$}}
    \RightLabel{(Constant$_+$)}
    \UnaryInfC{$\plusj{\varGamma}{\cst^o}{o}$}
    \DisplayProof
    \\\vspace{\baselineskip}
    \AxiomC{\phantom{$X^{Y^Z}$}}
    \RightLabel{$\Identity{+}$}
    \UnaryInfC{$\plusj{\varGamma, \evar^A \colon A}{\evar^A}{A}$}
    \DisplayProof
    \\\vspace{\baselineskip}
    \AxiomC{$\plusj{\varPi, \evar \colon A_0; \varSigma}{\expr}{A_1}$}
    \RightLabel{$\RuleLabel{+}{\to}{I}{}$}
    \UnaryInfC{$\plusj{\varPi; \varSigma}{\lamabs{\evar^{A_0}}{\expr}}{A_0 \to A_1}$}
    \DisplayProof
    \\\vspace{\baselineskip}
    \AxiomC{$\plusj{\varGamma}{\expr_0}{A_0 \to A_1}$}
    \AxiomC{$\plusj{\varGamma}{\expr_1}{A_0}$}
    \RightLabel{$\RuleLabel{+}{\to}{E}{}$}
    \BinaryInfC{$\plusj{\varGamma}{\app{\expr_0}{\expr_1}}{A_1}$}
    \DisplayProof
    \\\vspace{\baselineskip}
    \AxiomC{$\plusj{\varGamma}{\expr_0}{A_0}$}
    \AxiomC{$\plusj{\varGamma}{\expr_1}{A_1}$}
    \RightLabel{$\RuleLabel{+}{\wedge}{I}{}$}
    \BinaryInfC{$\plusj{\varGamma}{\pair{\expr_0}{\expr_1}}{A_0 \wedge A_1}$}
    \DisplayProof
    \\\vspace{\baselineskip}
    \AxiomC{$\plusj{\varGamma}{\expr}{A_0 \wedge A_1}$}
    \RightLabel{$\RuleLabel{+}{\wedge}{E}{0}$}
    \UnaryInfC{$\plusj{\varGamma}{\fst{\expr}}{A_0}$}
    \DisplayProof
    \\\vspace{\baselineskip}
    \AxiomC{$\plusj{\varGamma}{\expr}{A_0 \wedge A_1}$}
    \RightLabel{$\RuleLabel{+}{\wedge}{E}{1}$}
    \UnaryInfC{$\plusj{\varGamma}{\snd{\expr}}{A_1}$}
    \DisplayProof
   \end{center}
  \end{minipage}
  \:
  \begin{minipage}{.49\textwidth}
    \begin{center}
    \AxiomC{$\zeroj{\varPi; \varSigma, \evar \colon A}{\neutral}$}
    \RightLabel{$\Reductio{-}$}
    \UnaryInfC{$\minusj{\varPi; \varSigma}{\muabs{\evar^A}{\neutral}}{A}$}
    \DisplayProof
    \\\vspace{\baselineskip}
    \AxiomC{\phantom{$X^{Y^Z}$}}
    \RightLabel{(Constant$_{-}$)}
    \UnaryInfC{$\minusj{\varGamma}{\bullet^o}{o}$}
    \DisplayProof
    \\\vspace{\baselineskip}
    \AxiomC{\phantom{$X^{Y^Z}$}}
    \RightLabel{$\Identity{-}$}
    \UnaryInfC{$\minusj{\varGamma, \cvar^A \colon A}{\cvar^A}{A}$}
    \DisplayProof
    \\\vspace{\baselineskip}
    \AxiomC{$\minusj{\varPi; \varSigma, \cvar \colon A_1}{\cont}{A_0}$}
    \RightLabel{$\RuleLabel{-}{\gets}{I}{}$}
    \UnaryInfC{$\minusj{\varPi; \varSigma}{\lamabs{\cvar^{A_1}}{\cont}}{A_0 \gets A_1}$}
    \DisplayProof
    \\\vspace{\baselineskip}
    \AxiomC{$\minusj{\varGamma}{\cont_0}{A_0 \gets A_1}$}
    \AxiomC{$\minusj{\varGamma}{\cont_1}{A_1}$}
    \RightLabel{$\RuleLabel{-}{\gets}{E}{}$}
    \BinaryInfC{$\minusj{\varGamma}{\app{\cont_0}{\cont_1}}{A_0}$}
    \DisplayProof
    \\\vspace{\baselineskip}
    \AxiomC{$\minusj{\varGamma}{\cont_0}{A_0}$}
    \AxiomC{$\minusj{\varGamma}{\cont_1}{A_1}$}
    \RightLabel{$\RuleLabel{-}{\vee}{I}{}$}
    \BinaryInfC{$\minusj{\varGamma}{\pair{\cont_0}{\cont_1}}{A_0 \vee A_1}$}
    \DisplayProof
    \\\vspace{\baselineskip}
    \AxiomC{$\minusj{\varGamma}{\cont}{A_0 \vee A_1}$}
    \RightLabel{$\RuleLabel{-}{\vee}{E}{0}$}
    \UnaryInfC{$\minusj{\varGamma}{\fst{\cont}}{A_0}$}
    \DisplayProof
    \\\vspace{\baselineskip}
    \AxiomC{$\minusj{\varGamma}{\cont}{A_0 \vee A_1}$}
    \RightLabel{$\RuleLabel{-}{\vee}{E}{1}$}
    \UnaryInfC{$\minusj{\varGamma}{\snd{\cont}}{A_1}$}
    \DisplayProof
  \end{center}
\end{minipage}
\caption{A type system of $\BLC$.}\label{fig:bilateralcalculus}
\end{figure*}

Figure~\ref{fig:bilateralcalculus} shows the type system of $\BLC$
consisting of judgments $\plusj{\varGamma}{\expr}{A}$,
$\minusj{\varGamma}{\cont}{A}$, and $\zeroj{\varGamma}{\neutral}$,
where type environments
$\varGamma$ are defined as follows:
\begin{alignat*}{4}
  & \mbox{(type environments)} & \quad
  \varGamma & \Coloneqq \varPi; \varSigma & \qquad\quad
  \varPi & \Coloneqq \varnothing \mid \varPi, \evar \colon A & \qquad\quad
  \varSigma & \Coloneqq \varnothing \mid \varSigma, \cvar \colon A \enspace .
\end{alignat*}

Judgments $\plusj{\varPi; \varSigma}{\expr}{A}$, $\minusj{\varPi;
  \varSigma}{\cont}{A}$, and $\zeroj{\varPi; \varSigma}{\neutral}$
correspond to $\{\, \plusop{A} \mid A \in \varPi \,\}, \{\,
\minusop{A} \mid A \in \varSigma\,\} \vdash \plusop{A}$, $\{\,
\plusop{A} \mid A \in \varPi \,\}, \{\, \minusop{A} \mid A \in
\varSigma\,\} \vdash \minusop{A}$, and $\{\, \plusop{A} \mid A \in
\varPi \,\}, \{\, \minusop{A} \mid A \in \varSigma\,\} \vdash \bot$,
respectively


The type system contains rules about commands. Rule
$\Noncontradict$ defines a command from an expression and a
continuation. Additionally, even if a command occurs in a derivation,
the derivation does not necessarily end and may be continued by
$\Reductio{+}$ or $\Reductio{-}$.
The other inference rules about expressions are defined in a standard
manner.  The inference rules about continuations are defined
symmetrically to expressions.

Substitutions
$[\expr/\evar]$ and $[\cont/\cvar]$ (denoted by $\theta$) are
inductively defined in a standard component-wise and capture-avoiding
manner. We write
$\fev{\expr}$ and $\fev{\cont}$ for free expression variables in
$\expr$ and $\cont$, respectively. We also write $\fcv{\expr}$ and
$\fcv{\cont}$ for free continuation variables in $\expr$ and $\cont$,
respectively.

The bilateral $\lambda$-calculus is well designed. The so-called
weakening holds as follows:
\begin{proposition}
  \begin{enumerate}
  \item $\plusj{\varPi; \varSigma}{\expr}{A_0}$
    implies $\plusj{\varPi, \evar \colon A; \varSigma}{\expr}{A_0}$ and $\plusj{\varPi; \varSigma,
    \cvar \colon A}{\expr}{A_0}$,
  \item $\minusj{\varPi; \varSigma}{\cont}{A_0}$
    implies $\minusj{\varPi, \evar \colon A; \varSigma}{\cont}{A_0}$ and $\minusj{\varPi; \varSigma,
    \cvar \colon A}{\cont}{A_0}$, and
  \item $\zeroj{\varPi; \varSigma}{\neutral}$ implies $\zeroj{\varPi,
    \evar \colon A; \varSigma}{\neutral}$ and $\zeroj{\varPi;
    \varSigma, \cvar \colon A}{\neutral}$.
  \end{enumerate}
\end{proposition}
\begin{proof}
  By induction on derivation.
\end{proof}

The substitution lemma definitely holds as follows:
\begin{lemma}\label{lem:subst}
  \begin{enumerate}
  \item Assume $\plusj{\varPi, \evar \colon A_0;
    \varSigma}{\expr'}{A_1}$ and $\plusj{\varPi;
    \varSigma}{\expr}{A_0}$. Then, $\plusj{\varPi;
    \varSigma}{[\expr/\evar] \expr'}{A_1}$ holds.
  \item Assume $\minusj{\varPi, \evar \colon A_0;
    \varSigma}{\cont}{A_1}$ and $\plusj{\varPi;
    \varSigma}{\expr}{A_0}$. Then, $\minusj{\varPi; \varSigma}{[\expr
      /\evar]\cont}{A_1}$ holds.
  \item Assume $\zeroj{\varPi, \evar \colon A; \varSigma}{\neutral}$
    and $\plusj{\varPi; \varSigma}{\expr}{A}$. Then, $\zeroj{\varPi;
      \varSigma}{[\expr /\evar]\neutral}$ holds.
  \item Assume $\plusj{\varPi; \varSigma, \cvar \colon
    A_0}{\expr}{A_1}$ and $\minusj{\varPi;
    \varSigma}{\cont}{A_0}$. Then, $\plusj{\varPi; \varSigma}{[\cont
      /\cvar]\expr}{A_1}$ holds.
  \item Assume $\minusj{\varPi; \varSigma, \cvar \colon
    A_0}{\cont'}{A_1}$ and $\minusj{\varPi;
    \varSigma}{\cont}{A_0}$. Then, $\minusj{\varPi;
    \varSigma}{[\cont/\cvar] \cont'}{A_1}$ holds.
  \item Assume $\zeroj{\varPi; \varSigma, \cvar \colon A}{\neutral}$
    and $\minusj{\varPi; \varSigma}{\cont}{A}$. Then, $\zeroj{\varPi;
      \varSigma}{[\cont /\cvar]\neutral}$ holds.
  \end{enumerate}
\end{lemma}
\begin{proof}
  By induction on derivation.
\end{proof}

The bilateral $\lambda$-calculus enjoys the type uniqueness property, that
is, every expression and continuation has a unique positive and
negative type, respectively, as follows:
\begin{proposition}
  \begin{enumerate}
  \item If $\plusj{\varGamma}{\expr}{A_0}$ and
    $\plusj{\varGamma}{\expr}{A_1}$, then $A_0$ and
    $A_1$ are the same.
  \item If $\minusj{\varGamma}{\cont}{A_0}$ and
    $\minusj{\varGamma}{\cont}{A_1}$, then $A_0$ and $A_1$ are the
    same.
  \end{enumerate}
\end{proposition}
\begin{proof}
  The proposition holds immediately from the definition of the type system.
\end{proof}

\subsection{The Call-by-Value Lambda-Calculus $\cbvBLCeq$}\label{sec:cbvblc}

We define a call-by-value bilateral $\lambda$-calculus $\cbvBLCeq$. 
Types, expressions, continuations, commands, and typing rules
are the same as those of $\BLC$.
The values and the call-by-value evaluation contexts for expressions of $\cbvBLCeq$ are defined
as shown in Figure~\ref{fig:cbvvalueandcontext}.

\begin{figure}[t]
\begin{alignat*}{2}
  & \hbox{(values)}
  & \eValue & \Coloneqq
  \cst
  \mid \evar
  \mid \lamabs{\evar}{\expr}
  \mid \pair{\eValue}{\eValue}
  \mid \fst{\eValue}
  \mid \snd{\eValue}
  \mid \muabs{\cvar}{\bracket{\eValue}{\fst{\cvar}}}
  \mid \muabs{\cvar}{\bracket{\eValue}{\snd{\cvar}}}
  \\
  & \hbox{(contexts)}\;\;
  & \eEvalSymb & \Coloneqq
  \{-\}
  \mid \app{\eEvalSymb}{\expr}
  \mid \app{\eValue}{\eEvalSymb}
  \mid \pair{\eEvalSymb}{\expr}
  \mid \pair{\eValue}{\eEvalSymb}
  \mid \fst{\eEvalSymb}
  \mid \snd{\eEvalSymb}
  \enspace .
\end{alignat*}
\caption{Values and contexts of $\cbvBLCeq$.}\label{fig:cbvvalueandcontext}
\end{figure}


An evaluation context $\eEvalSymb$ is a expression with a hole $\{-\}$.
The expression obtained by filling the hole of $\eEvalSymb$ with an expression $\expr$
is denoted by $\eEval{\expr}$.
The equations of $\cbvBLCeq$ are shown in Figure~\ref{fig:cbvbieq}.

\begin{figure*}[t]
  \begin{minipage}[t]{.50\textwidth}
    \begin{align*}
      \app{(\lamabs{\evar}{\expr})}{\eValue} & \cbveq [\eValue/\evar] \expr\\
      \lamabs{\evar}{\app{\eValue}{\evar}} & \cbveq \eValue \tag*{if $\evar \not\in \fev{\eValue}$}\\
      \fst{\pair{\eValue_0}{\eValue_1}} & \cbveq \eValue_0\\
      \snd{\pair{\eValue_0}{\eValue_1}} & \cbveq \eValue_1\\
      \pair{\fst{\eValue}}{\snd{\eValue}} & \cbveq \eValue\\
      \muabs{\cvar}{\bracket{\expr}{\cvar}} & \cbveq \expr \tag*{if $\cvar \not\in \fcv{\expr}$}      
    \end{align*}
  \end{minipage}
  \begin{minipage}[t]{.49\textwidth}
    \begin{align*}
      \app{(\lamabs{\cvar}{\cont_0})}{\cont_1} & \cbveq [\cont_1/\cvar] \cont_0\\
      \lamabs{\cvar}{\app{\cont}{\cvar}} & \cbveq \cont \tag*{if $\cvar \not\in \fcv{\cont}$} \\
      \fst{\pair{\cont_0}{\cont_1}} & \cbveq \cont_0\\
      \snd{\pair{\cont_0}{\cont_1}} & \cbveq \cont_1\\
      \pair{\fst{\cont}}{\snd{\cont}} & \cbveq \cont\\
      \muabs{\evar}{\bracket{\evar}{\cont}} & \cbveq \cont \tag*{if $\evar \not\in \fev{\cont}$}
    \end{align*}
  \end{minipage}
\vspace{.1\baselineskip}
\begin{center}
  $\bracket{\eValue}{\muabs{\evar}{\neutral}} \cbveq [\eValue/\evar] N$
  \hspace{3cm}
  $\bracket{\muabs{\cvar}{\neutral}}{\cont} \cbveq [\cont/\cvar] N$
  \\[5pt]
  $\bracket{\eEval{\expr}}{\cont} \cbveq \bracket{\expr}{\muabs{\evar}{\bracket{\eEval{\evar}}{\cont}}}$
  \;\;
  \mbox{if $\evar$ is fresh}
\end{center}
  \caption{The equations of $\cbvBLCeq$.}\label{fig:cbvbieq}
\end{figure*}

Although careful readers will wonder why $\fst{\eValue}$ and
$\snd{\eValue}$ are values, they can often be  seen in
$\lambda$-calculi based on categorical semantics (cf.\ Definition 7.7
in Selinger's paper~\cite{Sel:concdc} and Figure 2 in
Wadler's paper~\cite{Wadler05}).
We also note that $\muabs{\cvar}{\bracket{\eValue}{\fst{\cvar}}}$ and 
$\muabs{\cvar}{\bracket{\eValue}{\snd{\cvar}}}$ are values for $A \vee B$, namely,
they mean the left and the right injections of $\eValue$, respectively.
We can define case expressions using pairs of continuations as follows:
\begin{center}
  $\inl{\expr} \equiv \muabs{\cvar}{\bracket{\expr}{\fst{\cvar}}} \qquad\qquad
  \inr{\expr}  \equiv \muabs{\cvar}{\bracket{\expr}{\snd{\cvar}}}$\\
  $\caseterm{\expr}{\evar_0}{\expr_0}{\evar_1}{\expr_1} \equiv \muabs{\cvar}{\bracket{\expr}{\pair{\muabs{\evar_0}{\bracket{\expr_0}{\cvar}}}{\muabs{\evar_1}{\bracket{\expr_1}{\cvar}}}}}$
\end{center}
\begin{center}
  \AxiomC{$\plusj{\varGamma}{\expr}{A_0}$}
  \UnaryInfC{$\plusj{\varGamma}{\inl{\expr}}{A_0 \vee A_1}$}
  \DisplayProof
  \qquad\qquad\qquad
  \AxiomC{$\plusj{\varGamma}{\expr}{A_1}$}
  \UnaryInfC{$\plusj{\varGamma}{\inr{\expr}}{A_0 \vee A_1}$}
  \DisplayProof
  \\\vspace{\baselineskip}
  \AxiomC{$\plusj{\varPi; \varSigma}{\expr}{A_0 \vee A_1}$}
  \AxiomC{$\plusj{\varPi, \evar_0 \colon A_0; \varSigma}{\expr_0}{A}$}
  \AxiomC{$\plusj{\varPi, \evar_1 \colon A_1; \varSigma}{\expr_1}{A}$}
  \TrinaryInfC{$\plusj{\varPi; \varSigma}{\caseterm{\expr}{\evar_0}{\expr_0}{\evar_1}{\expr_1}}{A}$}
  \DisplayProof
\end{center}
\vspace{-\baselineskip}
\begin{align*}
  & \bracket{\caseterm{\inl{\eValue}}{\evar_0}{\expr_0}{\evar_1}{\expr_1}}{\cont}\\
  & \equiv \bracket{\muabs{\cvar}{\bracket{\muabs{\cvar_2}{\bracket{\eValue}{\fst{\cvar_2}}}}{\pair{\muabs{\evar_0}{\bracket{\expr_0}{\cvar}}}{\muabs{\evar_1}{\bracket{\expr_1}{\cvar}}}}}}{\cont}\\
  & \cbveq \bracket{\muabs{\cvar_2}{\bracket{\eValue}{\fst{\cvar_2}}}}{\pair{\muabs{\evar_0}{\bracket{\expr_0}{\cont}}}{\muabs{\evar_1}{\bracket{\expr_1}{\cont}}}}\\
  & \cbveq \bracket{\eValue}{\fst{\pair{\muabs{\evar_0}{\bracket{\expr_0}{\cont}}}{\muabs{\evar_1}{\bracket{\expr_1}{\cont}}}}}
  \cbveq \bracket{\eValue}{\muabs{\evar_0}{\bracket{\expr_0}{\cont}}}
  \cbveq \bracket{[\eValue / \evar_0] \expr_0}{\cont} \enspace .
\end{align*}

  \begin{figure}[t]
  \begin{align*}
    & \mbox{(types)} &
    \dctype & \Coloneqq \chi \mid (\dctype \land \dctype) \mid (\dctype \vee \dctype) \mid (\negop{\dctype})
    \\
    & \mbox{(terms)} &
    \dcterm & \Coloneqq \dcvar \mid \dcpairr{\dcterm}{\dcterm} \mid \dcinl{\dcterm} \mid \dcinr{\dcterm} \mid \dcnotr{\dccoterm} \mid \dcabsr{\dccovar}{\dcstat}
    \\
    & \mbox{(coterms)} &
    \dccoterm & \Coloneqq \dccovar \mid \dcpairl{\dccoterm}{\dccoterm} \mid \dcfst{\dccoterm} \mid \dcsnd{\dccoterm} \mid \dcnotl{\dcterm} \mid \dcabsl{\dcvar}{\dcstat}
    \\
    & \mbox{(statements)} &
    \dcstat & \Coloneqq \dccut{\dcterm}{\dccoterm}
    \\
    & \mbox{(syntactical objects)} &
    \dcexpr & \Coloneqq \dcterm \mid \dccoterm \mid \dcstat
  \end{align*}
  \caption{The syntax of the dual calculus.}
  \label{fig:syntax_DC_short}
  \end{figure}

The calculus $\cbvBLCeq$ is $\cbvSubDCeq$ which is a sub-calculus of
an extension with the but-not connective of the call-by-value dual
calculus by Wadler~\cite{Wadler05}.
Types, terms, coterms, statements, and syntactical objects are shown
in Figure~\ref{fig:syntax_DC_short}.  A key difference from BLC is
that the dual calculus adopts $\neg$ as a primitive connective and
function types are syntactic sugar.
See Wadler's papers~\cite{Wadler05} or Appendix A for the details.
We can define a translation from $\cbvBLCeq$.  Consequently, the
consistency of our call-by-value calculus is obtained from the
consistency of the call-by-value dual calculus.
Specifically, we can obtain the following:
\begin{theorem}
  There exist translations $\fuga{-}$ from $\cbvBLCeq$ into
  $\cbvSubDCeq$ and $\gafu{-}$ from $\cbvSubDCeq$ into $\cbvBLCeq$,
  which satisfy:
  \begin{itemize}
  \item $D_0 \cbveq D_1$ implies $\fuga{D_0} \dcveq \fuga{D_1}$,
  \item $O_0 \dcveq O_1$ implies $\gafu{O_0} \cbveq \gafu{O_1}$,
  \item $\gafu{\fuga{D}} \cbveq D$ holds, and
  \item $\fuga{\gafu{O}} \dcveq O$ holds.
  \end{itemize}
  where $\dcveq$ is the equality relation of $\cbvSubDCeq$.
\end{theorem}
\begin{proof}
  See Appendix~\ref{sec:cbv_eq}.
\end{proof}

The theorem reasons about the call-by-value variant of BLC 
via the call-by-value dual calculus. 
Furthermore, the theorem shows that the but-not type $A\gets B$
in the call-by-value dual calculus is considered as
the function type for continuations. 
The theorem also reveals the difference between the dual calculus, 
whose negation type $\negop{A}$ is not involutive, and BLC, 
whose polarities $\plusop{A}$ and $\minusop{A}$ are involutive.

The negation type of the dual calculus can appear anywhere in a
type.  The negation type enables encoding of a coterm, say $\dccoterm$,
of type $\dctype$ to a term $\dcnotr{\dccoterm}$ of type
$\negop{\dctype}$, and handling of the encoded coterms as a part of
terms.
For instance, $\dcnotr{\dcpairl{K_1}{K_2}}$ of type $\negop{(\dctype_1 \vee \dctype_2)}$ is a term
which encodes the pair of coterms $K_1$ and $K_2$,
and functions, such as 
$\lambda\dcvar_2.\dcnotr{\dcabsl{\dcvar_1}{\dccut{\dcvar_2}{\dcnotl{\dcinl{\dcvar_1}}}}}$
of type $\negop{(\dctype_1\vee\dctype_2)} \to \negop{\dctype_1}$ that handles such terms,
are definable in the dual calculus.
The expressive power of BLC is strictly weaker than the dual calculus, 
since BLC does not permit defining such functions.
The theorem also raises a question whether 
BLC offers an adequate theoretical framework for expressing practical control operators. 
We conjecture that the polarities of BLC are enough for this purpose.
This is future work.


\section{Justifying the Duality of Functions}\label{sec:neutral}


In this section, we reason about the duality of functions in
Filinski's symmetric $\lambda$-calculus using the bilateral
$\lambda$-calculus.

\subsection{Filinski's Symmetric Lambda-Calculus}

A function of the type $A_0 \to A_1$ from expressions of the type
$A_0$ to expressions of the type $A_1$ can be regarded as a function
from continuations of the type $A_1$ to continuations of the type
$A_0$, and vice versa.  This property of functions is called the
duality of functions.

Filinski adopted the duality as a principle and constructed a
symmetric
$\lambda$-calculus~\cite{Filinski:89:DeclarativeContinuations:CTCS,Filinski:mthesis}. The
symmetric $\lambda$-calculus consists of functions $\func$,
expressions $\expr$, and continuations $\cont$. Functions consist of
$\lambda$-abstractions of expressions, decodings of expressions,
$\lambda$-abstractions of continuations, and decodings of
continuations as follows:
\begin{alignat*}{2}
  &
  \mbox{(functions)} & \qquad
\func^{A_0}_{A_1}
& \Coloneqq \Fielam{\eVar^{A_0}}{\expr_{A_1}}
\mid \edecode{\expr_{[A_0 \to A_1]}}
\mid \Ficlam{\cVar_{A_1}}{\cont^{A_0}}
\mid \cdecode{\cont^{[A_1 \gets A_0]}}
\enspace .
\end{alignat*}

Let $A_0 \to A_1$ be a function type. Filinski defined a function type
$[A_0 \to A_1]$ for an expression, which denotes an exponential object
${A_1}^{A_0}$ in categorical semantics, where $A_0$ and $A_1$ are
objects that correspond to types $A_0$ and $A_1$. We note that $A_2
\times A_0 \to A_1$ is bijective to $A_2 \to {A_1}^{A_0}$ in
categorical semantics.
Similarly, Filinski defined a function type $[A_1 \gets A_0]$ for a
continuation, which denotes a coexponential object ${A_0}_{A_1}$, and
$A_0 \to A_2 + A_1$ is bijective to ${A_0}_{A_1} \to A_2$.

Expressions and continuations consist of constants, variables,
applications of functions, and encodings of functions as follows:
\begin{alignat*}{3}
  & \mbox{(expressions)} &
  \expr_{o} & \Coloneqq
  \cst_{o}
  \mid \evar_{o}
  \mid \app{\func^{A}_{o}}{\expr_{A}} & \quad
  \expr_{[A_0 \to A_1]} & \Coloneqq
  \evar_{[A_0 \to A_1]}
  \mid \app{\func^{A}_{[A_0 \to A_1]}}{\expr_{A}}
  \mid \eencode{\func^{A_0}_{A_1}}\\
  & \mbox{(continuations)} & \;\;
  \cont^o & \Coloneqq
  \bullet^o
  \mid \cvar^o
  \mid \app{\func^{o}_A}{\cont^A}&
  \cont^{[A_1 \gets A_0]} & \Coloneqq
  \cvar^{[A_1 \gets A_0]}
  \mid \app{\func^{[A_1 \gets A_0]}_{A}}{\cont^{A}}
  \mid \cencode{\func^{A_0}_{A_1}}
  \enspace .
\end{alignat*}

We note that the encodings and decodings are defined to be
\emph{primitive} operators because the duality is adopted as a
principle.

We explain commands in Filinski's symmetric $\lambda$-calculus, which
is a triple called a \emph{configuration}:
\begin{center}
  \AxiomC{$\vdash \expr \colon \plusop{A_0}$}
  \AxiomC{$\vdash \func \colon A_0 \to A_1$}
  \AxiomC{$\vdash \cont \colon \negop{A_1}$}
  \TrinaryInfC{$\vdash \uatriple{\expr}{\func}{\cont}$}
  \DisplayProof
\end{center}
for the symmetric $\lambda$-calculus where $\vdash \expr \colon
\plusop{A_0}$ and $\vdash \cont \colon \negop{A_1}$ for expression $\expr$
of type $A_0$ and continuation $\cont$ of type $A_1$, respectively.
The notation was introduced by Ueda and Asai~\cite{uadiku}.
A difference from commands in the bilateral $\lambda$-calculus is that configurations are not pairs consisting of
expressions and continuations, but triples.
Another difference is that any configuration cannot be
abstracted by expression or continuation variables.
One other difference is that the continuation types are represented
using the negation connective in Filinski's calculus.


We can see that the configuration notion is also based on the duality principle.
If $\func$ is regarded as a function from expressions of the type
$A_0$ to expressions of the type $A_1$, then $\func$ is applied to
$\expr$ and an expression of the type $A_1$ that is consistent with
$\cont$ of the type $A_1$ is generated. Similarly, if $\func$ is
regarded as a function from continuations of the type $A_1$ to
continuations of the type $A_0$, then $\func$ is applied to $\cont$ and
a continuation of the type $A_0$ that is consistent with $\expr$ of the
type $A_0$ is generated. The configuration notion includes both
cases.

\subsection{Mutual Transformations between Functions}

Let us see how the duality occurs in the bilateral
$\lambda$-calculus.  The bilateral $\lambda$-calculus does not permit
anything neutral that is neither expression nor continuation.
Even if we want to define a neutral function, we must decide whether
the type of the function is either $\plusop{A_0 \to A_1}$ or
$\minusop{A_0 \gets A_1}$.
If we define a function between expressions which is applied to a
continuation, then the function cannot be \emph{as-is} applied to the
continuation, and vice versa.

However, we can define encodings $\cencode{\expr} =
\lamabs{\cvar}{\muabs{\evar}{\bracket{\app{\expr}{\evar}}{\cvar}}}$
and $\eencode{\cont} =
\lamabs{\evar}{\muabs{\cvar}{\bracket{\evar}{\app{\cont}{\cvar}}}}$ to
continuations and expressions in the bilateral $\lambda$-calculus,
respectively, and the encodings are mutual transformations as follows:
\begin{theorem}\label{thm:trans}
The following inferences are derivable:
\begin{center}
  \AxiomC{$\plusj{\varGamma}{\expr}{A_0 \to A_1}$}
  \UnaryInfC{$\minusj{\varGamma}{\cencode{\expr}}{A_0 \gets A_1}$}
  \DisplayProof
  \qquad\qquad
  \AxiomC{$\minusj{\varGamma}{\cont}{A_0 \gets A_1}$}
  \UnaryInfC{$\plusj{\varGamma}{\eencode{\cont}}{A_0 \to A_1}$}
  \DisplayProof
  \enspace .
\end{center}
\end{theorem}
\begin{proof}
See Appendix~\ref{sec:proofs}.
\end{proof}

The mutual transformations enjoy the following property:
\begin{theorem}\label{thm:mutualreduction}
  \begin{enumerate}
  \item $\bracket{\app{\eencode{\cont_0}}{\eValue}}{\cont_1} \cbveq
    \bracket{\eValue}{\app{\cont_0}{\cont_1}}$ holds,
  \item $\bracket{\eValue}{\app{\cencode{\expr}}{\cont}} \cbveq
    \bracket{\app{\expr}{\eValue}}{\cont}$ holds,
  \item $\bracket{\app{\eencode{\cencode{\expr}}}{\eValue}}{\cont} \cbveq
    \bracket{\app{\expr}{\eValue}}{\cont}$ holds, and
  \item $\bracket{\eValue}{\app{\cencode{\eencode{\cont_0}}}{\cont_1}}
    \cbveq \bracket{\eValue}{\app{\cont_0}{\cont_1}}$ holds.
  \end{enumerate}
\end{theorem}
\begin{proof}
  The first and second statements hold immediately from the definition
  of $\vred$, $\eencode{\cont}$, and $\cencode{\expr}$ as follows:
  \begin{align*}
    \bracket{\app{\eencode{\cont_0}}{\eValue}}{\cont_1}
    & \equiv \bracket{\app{(\lamabs{\evar}{\muabs{\cvar}{\bracket{\evar}{\app{\cont_0}{\cvar}}}})}{\eValue}}{\cont_1} 
    \cbveq \bracket{\muabs{\cvar}{\bracket{\eValue}{\app{\cont_0}{\cvar}}}}{\cont_1}
    \cbveq \bracket{\eValue}{\app{\cont_0}{\cont_1}}\\
    \bracket{\eValue}{\app{\cencode{\expr}}{\cont}}
    & \equiv \bracket{\eValue}{\app{(\lamabs{\cvar}{\muabs{\evar}{\bracket{\app{\expr}{\evar}}{\cvar}}})}{\cont}} 
    \cbveq \bracket{\eValue}{\muabs{\evar}{\bracket{\app{\expr}{\evar}}{\cont}}}
    \cbveq \bracket{\app{\expr}{\eValue}}{\cont} \enspace .
  \end{align*}
  The third and fourth statements hold from the the first
  and second statements.
\end{proof}

Theorems~\ref{thm:trans} and \ref{thm:mutualreduction} ensure
that we can always \emph{recover} to define functions between
expressions (and continuations) from functions between continuations
(resp.\ expressions) using the mutual transformations.
Thus, we confirm that the duality of functions is derived from
definability of the mutual transformations in the bilateral
$\lambda$-calculus.

\subsection{Dual Proofs for Functions}

We also provide an alternative justification of the duality using
derivation trees with proofs in their nodes introduced in
Section~\ref{sec:curry}.

We define a \emph{polarization}, which is a function from proof
variables and proof constants to expression or continuation variables
with types and expression or continuation constants, respectively. A
polarization for proofs is defined by
\begin{align*}
  p(\lamabs{\dvar}{\dobj}) & = \lamabs{p(\dvar)}{p(\dobj)} &
  p(\app{\dobj_0}{\dobj_1}) & = \app{p(\dobj_0)}{p(\dobj_1)} &
  p(\pair{\dobj_0}{\dobj_1}) & =   \pair{p(\dobj_0)}{p(\dobj_1)}\\
  p(\fst{\dobj}) & = \fst{p(\dobj)} &
  p(\snd{\dobj}) & = \snd{p(\dobj)} &
  p(\muabs{\dvar}{\dcom}) & = \muabs{p(\dvar)}{p(\dcom)}\\
  p(\bracket{\dobj_0}{\dobj_1}) & = \bracket{p(\dobj_0)}{p(\dobj_1)}
  \enspace .
\end{align*}

Let $V$ be a set of proof variables. We define $p(V)$ as the concatenation
of the positive type environments and the negative type environments
of $V$ by $p$.

Polarizations $p$ and $p'$ are equivalent if
\begin{itemize}
\item for any proof variable $\dvar$, $p(\dvar)$ and $p'(\dvar)$ have
  the same polarity, and
\item for any proof variables $\dvar$ and $\dvar'$, $p(\dvar) \equiv
  p(\dvar')$ implies $p'(\dvar) \equiv p'(\dvar')$, vice versa.
\end{itemize}

\begin{proposition}\label{prop:polequiv}
  Assume that $p$ and $p'$ are equivalent. Then,
  \begin{enumerate}
  \item $\plusj{p(V)}{p(\dobj)}{A}$ implies
    $\plusj{p'(V)}{p'(\dobj)}{A}$
  \item $\minusj{p(V)}{p(\dobj)}{A}$ implies
    $\minusj{p'(V)}{p'(\dobj)}{A}$, and
  \item $\zeroj{p(V)}{p(\dobj)}$ implies $\zeroj{p'(V)}{p'(\dobj)}$.
  \end{enumerate}
\end{proposition}

A polarization $p$ is a \emph{conjugate} of a polarization $p'$ if for
any variable $\dvar$, if $p(\dvar)$ is an expression variable, then
$p'(\dvar)$ is a continuation variable, and vice versa.

A derivation tree denoting a function has two proofs for
acceptance and rejection as follows:
\begin{theorem}
  \begin{enumerate}
  \item $\plusj{p(V)}{p(\dobj)}{A}$ implies that there exists a
    conjugate $p'$ of $p$ such that $\minusj{p'(V)}{p'(\dobj)}{A}$,
  \item $\minusj{p(V)}{p(\dobj)}{A}$ implies that there exists a
    conjugate $p'$ of $p$ such that $\plusj{p'(V)}{p'(\dobj)}{A}$, and
  \item $\zeroj{p(V)}{p(\dobj)}$ implies that there exists a conjugate
    $p'$ of $p$ such that $\zeroj{p'(V)}{p'(\dobj)}$.
  \end{enumerate}
\end{theorem}
\begin{proof}
  By induction on derivation. We note that
  Proposition~\ref{prop:polequiv} ensures differences between
  equivalent polarizations can be ignored.
\end{proof}

\subsection{A Short Remark about the Two Justifications}
Careful readers might think that
\begin{itemize}
\item no distinction of expression variables and continuation
  variables in derivation trees with proofs is better, and
\item BLC, which distinguishes expressions and continuations and
  requires the mutual transformations, is unnecessarily delicate.
\end{itemize}
However, BLC and the mutual transformations have an advantage in cases
that functions and arguments have common variables.  For example, a
function
$\lamabs{\dvar_2}{\muabs{\dvar_1}{\bracket{\dvar_0}{\dvar_2}}}$ cannot
be applied to an argument $\dvar_0$ under any assumption because the
function and argument must have converse polarities to each other.
Because the mutual transformations, which have no variable, can
respectively transform functions between expressions and continuations
to those between continuations and expressions in BLC,
$\app{\eencode{\lamabs{\cvar_2}{\muabs{\evar_1}{\bracket{\evar_0}{\cvar_2}}}}}{\evar_0}$
of the type $\plusop{A \to A}$ can be applied to $\evar_0$ of the type
$\plusop{A}$ where
$\lamabs{\cvar_2}{\muabs{\evar_1}{\bracket{\evar_0}{\cvar_2}}}$ has
the type $\minusop{A \gets A}$.

\section{Related Work and Discussion}\label{sec:related}

In this section, we discuss related work from three viewpoints of
symmetric $\lambda$-calculi on the formulae-as-types and approaches in
structural proof theory.

\subsection{Symmetric Lambda-Calculi}

The first symmetric $\lambda$-calculus was proposed by
Filinski~\cite{Filinski:89:DeclarativeContinuations:CTCS,Filinski:mthesis}.
Filinski described functions between continuations as follows:
\textit{``We can therefore equivalently view a function $f \colon A
  \to B$ as a continuation accepting a pair consisting of an $A$-type
  value and a $B$-accepting continuation. Such a pair will be called
  the context of a function application, and its type written as $[B
    \gets A]$''}.  In our observation on bilateralism, his intuition
is not only computationally but also proof-theoretic semantically
reasonable.  The underlying idea in defining our calculus is that
Filinski's $[A_1 \gets A_0]$ is regarded as $\minusop{A_0 \gets A_1}$.
We elaborate his idea in proof-theoretic semantics and carefully use
Rumfitt's polarities and the but-not connective, instead of simply
using the negation connective as Filinski did.

A symmetric $\lambda$-calculus proposed by Barbanera and Berardi was invented to
extract programs from classical logic proofs. Their calculus contains the involutive negation $A^\bot$
for each type $A$ and has symmetric application similar to commands in BLC.
The essential difference between their calculus and BLC is {\it polarity},
that is, the polarized type $\minusop{(A\vee B)}$ in BLC corresponds to $(A\vee B)^\bot$,
which is identified with $A^\bot \wedge B^\bot$ in their calculus. 
This lack of polarity information makes it difficult to reason about functions of Filinski's calculus. 


A calculus which was proposed by Lovas and Crary is the only symmetric
$\lambda$-calculus that corresponds to classical logic in which
expressions and continuations are symmetric on the
bilateralism~\cite{lovas2006}. They defined $\lambda$-terms similar to
those of the dual calculus which was defined by
Wadler~\cite{journals/sigplan/Wadler03}, and did not analyze the
duality of functions in Filinski's symmetric $\lambda$-calculus.  They
also did not adopt the implication $\to$ but the negation connective
$\neg$ as a primitive type constructor. An expression of function type
$A_0 \to A_1$ has type $\negop{(A_0 \wedge \negop{A_1})}$. Therefore,
it is necessary to use an inference rule that corresponds to reductio
ad absurdum in classical logic \emph{just} to define
$\lambda$-abstractions and applications of expressions in the simply
typed $\lambda$-calculus, unlike ours. We also show that the negative
polarity is suitable for representing continuations rather than the
negation connective on the notion of bilateralism.

Ueda and Asai investigated Filinski's symmetric $\lambda$-calculus,
and provided an explicit definition of commands by writing $\negop{A}$
for a continuation type $A$~\cite{uadiku}. However, they did not attempt to reason
about the neutrality of functions in the symmetric $\lambda$-calculus.
Also, the use of the negation connective to represent continuations is not
reasonable as we have shown in the present paper. Actually, they
also used the negation connective at only the \emph{outermost}
position of formulae. This operator of formulae should not be the
negation connective but the negative polarity on bilateralism.

Curien and Herbelin's $\CH$-calculus~\cite{CurienICFP00} 
corresponds to Gentzen's sequent calculus LK as well as the dual calculus.
This symmetric infrastructure, namely the duality of LK, exhibits the duality
between continuations and programs. 
Its symmetricity corresponds to that of the polarities in BLC
and to that of types $A$ and $\negop{A}$ in Ueda and Asai's calculus.
The calculi based on LK naturally contain the (not involutive) negation type, 
which provides a more expressive power than BLC, as noted in Section~\ref{sec:cbvblc}.
This observation raises an interesting question: 
What is the role of the negation type in practical programming languages?

\subsection{Approaches in Structural Proof Theory}

%
%

Girard and Parigot constructed calculi corresponding to classical
logic~\cite{Girard1991,Parigotlpar92} and analyzed classical logic
proof-theoretically. Girard also invented linear logic~\cite{Gir87a},
which is very useful for analyzing classical logic. Danos et
al.\ confirmed that classical logic has well behaved fragments using
the positive and negative polarities~\cite{LKTLKQ,DJS1997}. The
calculi invented through their approaches are larger than or
incomparable to ours because their motivations are different from
ours.  A goal of our work is not to analyze classical logic but to
construct a minimal calculus to justify the duality of functions and
the computations that delimited continuations raise.  Although
analyzing negations is a topic of great interest in proof
theory~\cite{Nel:conf,Geach1965,McCall1967,Clark1978,Dosen1984,dunn:93,Dummett1996,Avronnegation,restall:introduction-to-substructural-logics,humberstone_il:2000a,journals/sLogica/DunnZ05},
we investigated the negation-free fragment of bilateral natural
deduction.

Dual intuitionistic logic, which is symmetric to intuitionistic logic,
is well known in structural proof
theory~\cite{Goodman,Urbas,journals/sLogica/Shramko05}. A combined
logic of intuitionistic and dual intuitionistic logics is classical
logic.  Whereas most of the logics are based on sequent calculi,
Wansing constructed a natural deduction that
can perform verification and falsification that corresponds to proving
$\plusop{A}$ and $\minusop{A}$, respectively, in our calculus~\cite{journals/logcom/Wansing16}.
However, a series of his works analyzed \emph{refutation}, which is a
proof for falsification in the context of studying various negations
as seen in structural proof
theory~\cite{conf/aiml/Wansing04,Wansingaiml,journals/logcom/Wansing16}. This
is different from the objective in the present paper.  He also neither
provided a $\lambda$-calculus based on bilateralism nor described
computational aspects, such as continuation controls.
Tranchini also constructed a natural deduction of dual intuitionistic
logic~\cite{journals/sLogica/Tranchini12}.  Our calculus seems to
correspond to a negation-free fragment of his natural deduction.


\section{Conclusion and Future Work}\label{sec:conclusion}

In this paper, we proposed a symmetric $\lambda$-calculus called the
bilateral $\lambda$-calculus with the but-not connective based on
bilateralism in proof-theoretic semantics.
The formulae-as-types notion was extended to consider Rumfitt's
reductio, which corresponds to reductio ad absurdum as a
$\mu$-abstraction of a first-class command in our calculus.  Its
call-by-value calculus can be defined as a sub-calculus of Wadler's
call-by-value dual calculus.
%
We showed that the duality of functions is derived from definability
of the mutual transformations between expressions and continuations in
the bilateral $\lambda$-calculus. We also showed that every typable
function has dual types.
%

In this paper, we have provided a method to justify a few notions in
the theory of $\lambda$-calculi on bilateralism.
The bilateral analysis in this paper targets the duality of functions
in Filinski's symmetric $\lambda$-calculus.
%
%
Bilateral analyses of asymmetric calculi constitute our future work.
%

The call-by-value variant of BLC corresponds to a sub-calculus of the
call-by-value dual calculus with the but-not connective. It is also
future work to clarify what practical uses are derived from the
difference between BLC and the dual calculus.




\newpage
\bibliographystyle{plainurl}
\bibliography{draft}

\begin{thebibliography}{10}

\bibitem{Avronnegation}
Arnon Avron.
\newblock Negation: Two points of view.
\newblock In {\em What is Negation?}, Applied Logic Series, pages 3--22. 1999.

\bibitem{Clark1978}
Keith~L. Clark.
\newblock Negation as failure.
\newblock In {\em Logic and Data Bases}, pages 293--322. Plenum Press, 1978.

\bibitem{CurienICFP00}
Pierre-Louis Curien and Hugo Herbelin.
\newblock The duality of computation.
\newblock In {\em Proc.\ ICFP}, pages 233--243, 2000.

\bibitem{curry}
Haskell.~B. Curry.
\newblock Functionality in combinatory logic.
\newblock In {\em Proc.\ the National Academy of Sciences of USA}, volume~20,
  pages 584--590, 1934.

\bibitem{LKTLKQ}
Vincent Danos, Jean-Baptiste Joinet, and Harold Schellinx.
\newblock {LKQ} and {LKT}: Sequent calculi for second order logic based upon
  dual linear decompositions of classical implication.
\newblock In {\em Proceedings of the Workshop on Advances in Linear Logic},
  pages 211--224, 1995.

\bibitem{DJS1997}
Vincent Danos, Jean-Baptiste Joinet, and Harold Schellinx.
\newblock A new deconstructive logic: Linear logic.
\newblock {\em Journal of Symbolic Logic}, 62(3):755--807, 1997.

\bibitem{Dosen1984}
Kosta Do\u{s}en.
\newblock Negative modal operators in intuitionistic logic.
\newblock {\em Publications de l'Institut Math\'{e}matique}, 35(49):3--14,
  1984.

\bibitem{Dum:logbm}
Michael Dummett.
\newblock {\em The Logical Basis of Metaphysics}.
\newblock Duckworth, 1991.

\bibitem{Dummett1996}
Michael Dummett.
\newblock {\em The Seas of Language}.
\newblock Oxford University Press, 1996.

\bibitem{dunn:93}
Jon~Michael Dunn.
\newblock Star and perp: Two treatments of negation.
\newblock {\em Philisophical Perspectives}, 5:331--357, 1993.

\bibitem{journals/sLogica/DunnZ05}
Jon~Michael Dunn and Chunlai Zhou.
\newblock Negation in the context of gaggle theory.
\newblock {\em Studia Logica}, 80(2--3):235--264, 2005.

\bibitem{Filinski:89:DeclarativeContinuations:CTCS}
Andrzej Filinski.
\newblock Declarative continuations: An investigation of duality in programming
  language semantics.
\newblock In {\em Proc.\ CTCS}, volume 389 of {\em LNCS}, pages 224--249, 1989.

\bibitem{Filinski:mthesis}
Andrzej Filinski.
\newblock Declarative continuations and categorical duality.
\newblock Master's thesis, {DIKU} Computer Science Department, University of
  Copenhagen, 1989.

\bibitem{Geach1965}
Peter~T. Geach.
\newblock Assertion.
\newblock {\em The Philosophical Review}, 74:449--465, 1965.

\bibitem{Gentzen1}
Gerhard Karl~Erich Gentzen.
\newblock Untersuchungen {\"u}ber das logische schlie{\ss}en.
\newblock {\em Mathematische Zeitschrift}, 39:176--210, 1934.

\bibitem{Gentzen2}
Gerhard Karl~Erich Gentzen.
\newblock Untersuchungen {\"u}ber das logische schlie{\ss}en.
\newblock {\em Mathematische Zeitschrift}, 39:405--431, 1935.

\bibitem{Gir87a}
Jean-Yves Girard.
\newblock Linear logic.
\newblock {\em Theoretical Computer Science}, 50:1--102, 1987.

\bibitem{Girard1991}
Jean-Yves Girard.
\newblock A new constructive logic: classic logic.
\newblock {\em Mathematical Structures in Computer Science}, 1(3):255--296,
  1991.

\bibitem{Goodman}
Nicolas~D. Goodman.
\newblock The logic of contradiction.
\newblock {\em Zeitschrift f{\"u}r mathematische Logik und Grundlagen der
  Mathematik}, 27:119--126, 1981.

\bibitem{conf/aiml/GorePT08}
Rajeev Gor{\'e}, Linda Postniece, and Alwen Tiu.
\newblock Cut-elimination and proof-search for bi-intuitionistic logic using
  nested sequents.
\newblock In {\em Advances in Modal Logic}, pages 43--66, 2008.

\bibitem{Griffin90}
Timothy~G. Griffin.
\newblock A formulae-as-types notion of control.
\newblock In {\em Proc.\ POPL}, pages 47--58, 1990.

\bibitem{HowardW:fortnc}
William.~A. Howard.
\newblock The formulae-as-types notion of construction.
\newblock In {\em Essays on Combinatory Logic, Lambda Calculus, and Formalism},
  pages 479--490. Academic Press, 1980.

\bibitem{humberstone_il:2000a}
Lloyd Humberstone.
\newblock The revival of rejective negation.
\newblock {\em Journal of Philosophical Logic}, 29(4):331--381, 2000.

\bibitem{journals/jphil/Kurbis16}
Nils K{\"u}rbis.
\newblock Some comments on ian rumfitt's bilateralism.
\newblock {\em Journal of Philosophical Logic}, 45(6):623--644, 2016.

\bibitem{lovas2006}
William Lovas and Karl Crary.
\newblock Structural normalization for classical natural deduction.
\newblock Manuscript, 2006.
\newblock URL: \url{http://www.cs.cmu.edu/{\~{}}wlovas/papers/clnorm.pdf}.

\bibitem{McCall1967}
Storrs McCall.
\newblock Contrariety.
\newblock {\em Notre Dame Journal of Formal Logic}, 8:121--138, 1967.

\bibitem{Nel:conf}
David Nelson.
\newblock Constructible falsity.
\newblock {\em Journal of Symbolic Logic}, 14(2):16--26, 1949.

\bibitem{Parigotlpar92}
Michel Parigot.
\newblock $\lambda$$\mu$-calculus: An algorithmic interpretation of classical
  natural deduction.
\newblock In {\em Proc.\ LPAR}, volume 624 of {\em LNAI}, pages 190--201, 1992.

\bibitem{prior60analysis}
Arthur~N. Prior.
\newblock The runabout inference-ticket.
\newblock {\em Analysis}, 21(2):38--39, 1960.

\bibitem{oai:CiteSeerPSU:230037}
Greg Restall.
\newblock Extending intuitionistic logic with subtraction, 1997.

\bibitem{restall:introduction-to-substructural-logics}
Greg Restall.
\newblock {\em An Introduction to Substructural Logics}.
\newblock Routledge, 2000.

\bibitem{rumfitt_i:2000a}
Ian Rumfitt.
\newblock {``Y}es'' and ``no{''}.
\newblock {\em Mind}, 109(477):781--823, 2000.

\bibitem{Sel:concdc}
Peter Selinger.
\newblock Control categories and duality: On the categorical semantics of the
  lambda-mu calculus.
\newblock {\em Mathematical Structures in Computer Science}, 11(2):207--260,
  2001.

\bibitem{journals/sLogica/Shramko05}
Yaroslav Shramko.
\newblock Dual intuitionistic logic and a variety of negations: The logic of
  scientific research.
\newblock {\em Studia Logica}, 80(2--3):347--367, 2005.

\bibitem{journals/sLogica/Tranchini12}
Luca Tranchini.
\newblock Natural deduction for dual-intuitionistic logic.
\newblock {\em Studia Logica}, 100(3):631--648, 2012.

\bibitem{uadiku}
Yayoi Ueda and Kenichi Asai.
\newblock Reinvestigation of symmetric lambda calculus.
\newblock In {\em Proc.\ the 4th DIKU-IST Joint Workshop on Foundations of
  Software}, pages 10--26, 2011.

\bibitem{Urbas}
Igor Urbas.
\newblock Dual-intuitionistic logic.
\newblock {\em Notre Dame Journal of Formal Logic}, 37(3):440--451, 1996.

\bibitem{journals/sigplan/Wadler03}
Philip Wadler.
\newblock Call-by-value is dual to call-by-name.
\newblock In {\em Proc.\ ICFP}, pages 189--201, 2003.

\bibitem{Wadler05}
Philip Wadler.
\newblock Call-by-value is dual to call-by-name, reloaded.
\newblock In {\em Proc.\ RTA}, volume 3467 of {\em LNCS}, pages 185--203, 2005.

\bibitem{conf/aiml/Wansing04}
Heinrich Wansing.
\newblock Connexive modal logic.
\newblock In {\em Proc.\ AIML}, pages 367--383, 2004.

\bibitem{Wansingaiml}
Heinrich Wansing.
\newblock Proofs, disproofs, and their duals.
\newblock {\em Advances in Modal Logic}, 8:483--505, 2010.

\bibitem{journals/logcom/Wansing16}
Heinrich Wansing.
\newblock Falsification, natural deduction and bi-intuitionistic logic.
\newblock {\em Journal of Logic and Compututation}, 26(1):425--450, 2016.

\end{thebibliography}

\newpage


\appendix

\section{The Call-by-Value Calculus of The Bilateral Lambda-Calculus}\label{sec:cbv_eq}

%
We introduce a call-by-value
strategy to BLC and define a computationally consistent call-by-value calculus,
which is equivalent to
a sub-calculus of the call-by-value dual calculus by Wadler~\cite{Wadler05} without negation.
Consequently, the consistency of our call-by-value calculus is obtained 
from the consistency of the call-by-value dual calculus. 


\subsection{The Call-by-Value Dual Calculus $\cbvSubDCeq$}

This subsection compares $\cbvBLCeq$ with dual calculus invented by Wadler~\cite{journals/sigplan/Wadler03,Wadler05}, 
which corresponds to the classical sequent calculus
on the notion of formulae-as-types.
The call-by-value calculus of the dual calculus is known as a well established and computationally consistent system
because it has the so-called \emph{CPS-semantics}~\cite{journals/sigplan/Wadler03}
and is equivalent to the call-by-value $\lambda\mu$-calculus~\cite{Wadler05}. 
We will show that our $\cbvBLCeq$ is equivalent to a sub-calculus of the call-by-value dual calculus
by giving an isomorphism between them. 

We first recall the dual calculus. 
Suppose that countable sets of type variables, term variables, and coterm variables are given. 
Let $\chi$, $\dcvar$, and $\dccovar$ range over type variables, term variables, and coterm variables, respectively.
Types, terms, coterms, statements, and syntactical objects are summarized in Figure~\ref{fig:syntax_DC}. 
Substitution $[\dcterm/\dcvar]\dcexpr$ of $\dcvar$ in an expression $\dcexpr$
for $\dcterm$ is defined in a standard component-wise and capture-avoiding
manner.
Similarly, substitution $[\dccoterm/\dccovar]\dcexpr$ is also defined.

  \begin{figure}[t]
  \begin{align*}
    & \mbox{(types)} &
    \dctype & \Coloneqq \chi \mid (\dctype \land \dctype) \mid (\dctype \vee \dctype) \mid (\negop{\dctype})
    \\
    & \mbox{(terms)} &
    \dcterm & \Coloneqq \dcvar \mid \dcpairr{\dcterm}{\dcterm} \mid \dcinl{\dcterm} \mid \dcinr{\dcterm} \mid \dcnotr{\dccoterm} \mid \dcabsr{\dccovar}{\dcstat}
    \\
    & \mbox{(coterms)} &
    \dccoterm & \Coloneqq \dccovar \mid \dcpairl{\dccoterm}{\dccoterm} \mid \dcfst{\dccoterm} \mid \dcsnd{\dccoterm} \mid \dcnotl{\dcterm} \mid \dcabsl{\dcvar}{\dcstat}
    \\
    & \mbox{(statements)} &
    \dcstat & \Coloneqq \dccut{\dcterm}{\dccoterm}
    \\
    & \mbox{(syntactical objects)} &
    \dcexpr & \Coloneqq \dcterm \mid \dccoterm \mid \dcstat
  \end{align*}
  \caption{The syntax of the dual calculus.}
  \label{fig:syntax_DC}
  \end{figure}

A judgment has the form of
$\varGamma \vdash \varDelta \mid \dcterm \colon \dctype$,
$\varGamma \mid \dcstat \vdash \varDelta$, or
$\dccoterm \colon \dctype \mid \varGamma \vdash \varDelta$,
where
$\varGamma$ is a type environment for terms that is a finite set of the form $\dcvar\colon\dctype$ and
$\varDelta$ is a type environment for coterms that is a finite set of the form $\dccovar\colon\dctype$.
Figure~\ref{fig:typing_DC} shows the inference rules.

  \begin{figure}[t]
    \begin{center}
      \AxiomC{
        $\varGamma \vdash \varDelta \mid \dcterm \colon \dctype$
      }
      \AxiomC{
        $\dccoterm \colon \dctype \mid \varGamma \vdash \varDelta$
      }
      \BinaryInfC{
        $\varGamma \mid \dccut{\dcterm}{\dccoterm} \vdash \varDelta$
      }
      \DisplayProof
      \\\vspace{\baselineskip}
      \AxiomC{}
      \UnaryInfC{
        $\varGamma,\dcvar \colon \dctype \vdash \varDelta \mid \dcvar \colon \dctype$
      }
      \DisplayProof
      \qquad 
      \AxiomC{}
      \UnaryInfC{
        $\dccovar \colon \dctype \mid \varGamma \vdash \varDelta, \dccovar \colon \dctype$
      }
      \DisplayProof
      \\[10pt]
      \AxiomC{
        $\varGamma \vdash \varDelta \mid \dcterm_1 \colon \dctype_1$
      }
      \AxiomC{
        $\varGamma \vdash \varDelta \mid \dcterm_2 \colon \dctype_2$
      }
      \BinaryInfC{
        $\varGamma \vdash \varDelta \mid \dcpairr{\dcterm_1}{\dcterm_2} \colon \dctype_1\land\dctype_2$
      }
      \DisplayProof
      \\\vspace{\baselineskip}
      \AxiomC{
        $\dccoterm \colon \dctype_1 \mid \varGamma \vdash \varDelta$
      }
      \UnaryInfC{
        $\dcfst{\dccoterm} \colon \dctype_1\land\dctype_2 \mid \varGamma \vdash \varDelta$
      }
      \DisplayProof
      \qquad
      \AxiomC{
        $\dccoterm \colon \dctype_2 \mid \varGamma \vdash \varDelta$
      }
      \UnaryInfC{
        $\dcsnd{\dccoterm} \colon \dctype_1\land\dctype_2 \mid \varGamma \vdash \varDelta$
      }
      \DisplayProof
      \\\vspace{\baselineskip}
      \AxiomC{
        $\varGamma \vdash \varDelta \mid \dcterm \colon \dctype_1$
      }
      \UnaryInfC{
        $\varGamma \vdash \varDelta \mid \dcinl{\dcterm} \colon \dctype_1\vee\dctype_2$
      }
      \DisplayProof
      \qquad
      \AxiomC{
        $\varGamma \vdash \varDelta \mid \dcterm \colon \dctype_2$
      }
      \UnaryInfC{
        $\varGamma \vdash \varDelta \mid \dcinr{\dcterm} \colon \dctype_1\vee\dctype_2$
      }
      \DisplayProof
      \\\vspace{\baselineskip}
      \AxiomC{
        $\dccoterm_1 \colon \dctype_1 \mid \varGamma \vdash \varDelta$
      }
      \AxiomC{
        $\dccoterm_2 \colon \dctype_2 \mid \varGamma \vdash \varDelta$
      }
      \BinaryInfC{
        $\dcpairl{\dccoterm_1}{\dccoterm_2} \colon \dctype_1\vee\dctype_2 \mid \varGamma \vdash \varDelta$
      }
      \DisplayProof
      \\\vspace{\baselineskip}
      \AxiomC{
        $\dccoterm \colon \dctype \mid \varGamma \vdash \varDelta$
      }
      \UnaryInfC{
        $\varGamma \vdash \varDelta \mid \dcnotr{\dccoterm} \colon \negop{\dctype}$
      }
      \DisplayProof
      \qquad
      \AxiomC{
        $\varGamma \vdash \varDelta \mid \dcterm \colon \dctype$
      }
      \UnaryInfC{
        $\dcnotl{\dcterm} \colon \negop{\dctype} \mid \varGamma \vdash \varDelta$
      }
      \DisplayProof
      \\\vspace{\baselineskip}
      \AxiomC{
        $\varGamma \mid \dcstat \vdash \varDelta, \dccovar \colon \dctype$
      }
      \UnaryInfC{
        $\varGamma \vdash \varDelta \mid \dcabsr{\dccovar}{\dcstat} \colon \dctype$
      }
      \DisplayProof
      \qquad
      \AxiomC{
        $\dcvar \colon \dctype, \varGamma \mid \dcstat \vdash \varDelta$
      }
      \UnaryInfC{
        $\dcabsl{\dcvar}{\dcstat} \colon \dctype \mid \varGamma \vdash \varDelta$
      }
      \DisplayProof
    \end{center}
    \caption{The inference rules of the dual calculus.}
    \label{fig:typing_DC}
  \end{figure}

  We then recall the call-by-value calculus of the dual calculus.
  The values and the call-by-value evaluation contexts are
  defined as follows:
  \begin{alignat*}{2}
    &\hbox{(values)}&
    \dcValue & \Coloneqq 
    \dcvar
    \mid \dcpairr{\dcValue}{\dcValue}
    \mid \dcinl{\dcValue}
    \mid \dcinr{\dcValue}
    \mid \dcnotr{\dccoterm}\\
    & & &
    \;\;\mid\; \dcabsr{\dccovar}{\dccut{\dcValue}{\dcfst{\dccovar}}}
    \mid \dcabsr{\dccovar}{\dccut{\dcValue}{\dcsnd{\dccovar}}}
    \\
    &\hbox{(contexts)}\quad&
    \dcEvalSymb & \Coloneqq 
    \{-\} 
    \mid \dcpairr{\dcEvalSymb}{\dcterm}
    \mid \dcpairr{\dcValue}{\dcEvalSymb}
    \mid \dcinl{\dcEvalSymb}
    \mid \dcinr{\dcEvalSymb}
  \end{alignat*}
  
  \begin{figure}[t]
    \begin{center}
      \begin{alignat*}{2}
        & (\beta\mathord\land_0)
        & \;\;
        \dccut{\dcpairr{\dcValue_0}{\dcValue_1}}{\dcfst{\dccoterm}}
        & \dcveq
        \dccut{\dcValue_0}{\dccoterm}
        \\
        & (\beta\land_1)
        &
        \dccut{\dcpairr{\dcValue_0}{\dcValue_1}}{\dcsnd{\dccoterm}}
        & \dcveq
        \dccut{\dcValue_1}{\dccoterm}
        \\
        & (\beta\vee_0)
        &
        \dccut{\dcinl{\dcValue}}{\dcpairl{\dccoterm_0}{\dccoterm_1}}
        & \dcveq
        \dccut{\dcValue}{\dccoterm_0}
        \\
        & (\beta\vee_1)
        &
        \dccut{\dcinr{\dcValue}}{\dcpairl{\dccoterm_0}{\dccoterm_1}}
        & \dcveq
        \dccut{\dcValue}{\dccoterm_1}
        \\
        & (\beta\neg)
        &
        \dccut{\dcnotr{\dccoterm}}{\dcnotl{\dcterm}}
        & \dcveq
        \dccut{\dcterm}{\dccoterm}
        \\
        & (\beta R)
        &
        \dccut{\dcValue}{\dcabsl{\dcvar}{\dcstat}}
        & \dcveq
        [\dcValue/\dcvar]\dcstat
        \\
        & (\beta L)
        &
        \dccut{\dcabsr{\dccovar}{\dcstat}}{\dccoterm}
        & \dcveq
        [\dccoterm/\dccovar]\dcstat
\\
        & (\eta R)
        & \;\;
        \dcterm
        & \dcveq
        \dcabsr{\dccovar}{\dccut{\dcterm}{\dccovar}}
         \tag*{if $\dccovar$ is fresh}
        \\
        & (\eta L)
        &
        \dccoterm
        & \dcveq
        \dcabsl{\dcvar}{\dccut{\dcvar}{\dccoterm}}
        \tag*{if $\dcvar$ is fresh}
        \\         
        & (\eta\land)
        &
        \dcValue
        & \dcveq
        \dcpairr{
          \dcabsr{\dccovar}{\dccut{\dcValue}{\dcfst{\dccovar}}}
        }{
          \dcabsr{\dccovar}{\dccut{\dcValue}{\dcsnd{\dccovar}}}
        } \tag*{if $\dccovar$ is fresh}
        \\
        & (\eta\vee)
        &
        \dccoterm
        & \dcveq
        \dcpairl{
          \dcabsl{\dcvar}{\dccut{\dcinl{\dcvar}}{\dccoterm}}
        }{
          \dcabsl{\dcvar}{\dccut{\dcinr{\dcvar}}{\dccoterm}}
        } \tag*{if $\dcvar$ is fresh}
        \\
        & (\eta\neg)
        &
        \dcValue
        & \dcveq
        \dcnotr{
          \dcabsl{\dcvar}{
            \dccut{\dcValue}{\dcnotl{\dcvar}}
          }
        } \tag*{if $\dcvar$ is fresh}
        \\
        & (\zeta) &
        \dccut{\dcEval{\dcterm}}{\dccoterm}
        & \dcveq
        \dccut{\dcterm}{
          \dcabsl{\dcvar}{\dccut{\dcEval{\dcvar}}{\dccoterm}}
        }
        \tag*{if $\dcvar$ is fresh}
      \end{alignat*}
    \end{center}
    \caption{The equations of the call-by-value dual calculus.}
    \label{fig:cbveq_DC}
  \end{figure}

  Figure~\ref{fig:cbveq_DC} presents the call-by-value equation $\dcveq$ of the dual calculus.
  In the call-by-value dual calculus, the implication type $\dctype_0\dcimp\dctype_1$ with its term $\lambda x.M$, coterm $\dcappl{\dcterm}{\dccoterm}$
  can be defined as
  \begin{align*}
    \dctype_0\dcimp\dctype_1
    &\equiv
    \negop{(\dctype_0\land\negop{\dctype_1})}
    &
    \lambda x.M
    &\equiv
    \dcnotr{\dcabsl{\dcvar'}{\dccut{\dcvar'}{\dcfst{\dcabsl{x}{\dccut{\dcvar'}{\dcsnd{\dcnotl{M}}}}}}}}
    \\
    \dcappl{\dcterm}{\dccoterm}
    &\equiv
    \dcnotl{\dcpairr{\dcterm}{\dcnotr{\dccoterm}}}
  \end{align*}
by using $\land$, $\vee$, and $\neg$ as follows:
\begin{proposition}\label{prop:dcimp}
  The following inferences are derivable:
  \begin{center}
    \AxiomC{
      $\varGamma, \dcvar \colon \dctype_0 \vdash \varDelta \mid \dcterm \colon \dctype_1$
    }
    \UnaryInfC{
      $\varGamma \vdash \varDelta \mid \lambda\dcvar.\dcterm \colon \dctype_0\dcimp\dctype_1$
    }
    \DisplayProof
    \qquad
    \AxiomC{
      $\varGamma \vdash \varDelta \mid M \colon \dctype_0$
    }
    \AxiomC{
      $K \colon \dctype_1 \mid \varGamma \vdash \varDelta$
    }    
    \BinaryInfC{
      $\dcappl{\dcterm}{\dccoterm} \colon \dctype_0\dcimp\dctype_1 \mid \varGamma \vdash \varDelta$
    }
    \DisplayProof
    \enspace .
  \end{center}
  Also, $\lambda x.M$ is a value and the following equations hold:
  \begin{alignat*}{2}
    & (\beta\mathord\dcimp)\qquad &
    \dccut{(\lambda x.M_0)}{(\dcappl{M_1}{K})}
    &\dcveq
    \dccut{M_1}{\dcabsl{x}{\dccut{M_0}{K}}}
    \\
    & (\eta\mathord\to) &
    \dcValue
    &\dcveq
    \lambda \dcvar.
    (\dcabsr{
      \dccovar
    }{
      \dccut{\dcValue}{(\dcappl{\dcvar}{\dccovar})}
    })
    \enspace .
  \end{alignat*}
\end{proposition}
\begin{proof}
  The inference part is shown immediately by the definition of $\lambda x.M$ and $\dcappl{M}{K}$.
  The term $\lambda x.M$ is a value, since it has the form $\dcnotr{K}$.
  The equation $(\beta\mathord{\dcimp})$ is shown by the case analysis of $M_1$. 

  (a) If $M_1$ is not a value, then the claim is obtained by using $(\nu\wedge_0)$: 
  \begin{align*}
  & \dccut{(\lambda x.M_0)}{(\dcappl{M_1}{K})}\\
  &\equiv
  \dccut{
    \dcnotr{\dcabsl{\dcvar'}{\dccut{\dcvar'}{\dcfst{\dcabsl{x}{\dccut{\dcvar'}{\dcsnd{\dcnotl{M_0}}}}}}}}
  }{
    \dcnotl{\dcpairr{M_1}{\dcnotr{K}}}
  }
  \\
  &\dcveq
  \dccut{
    \dcpairr{M_1}{\dcnotr{K}}
  }{
    \dcabsl{\dcvar'}{\dccut{\dcvar'}{\dcfst{\dcabsl{x}{\dccut{\dcvar'}{\dcsnd{\dcnotl{M_0}}}}}}}
  }
  \\
  &\dcveq
  \dccut{M_1}{
    \dcabsl{\dcvar''}{
      \dccut{
        \dcpairr{\dcvar''}{\dcnotr{K}}
      }{
        \dcabsl{\dcvar'}{\dccut{\dcvar'}{\dcfst{\dcabsl{x}{\dccut{\dcvar'}{\dcsnd{\dcnotl{M_0}}}}}}}
      }
    }
  }
  \\
  &\dcveq
  \dccut{M_1}{
    \dcabsl{\dcvar''}{
      \dccut{\dcpairr{\dcvar''}{\dcnotr{K}}}{\dcfst{\dcabsl{x}{\dccut{\dcpairr{\dcvar''}{\dcnotr{K}}}{\dcsnd{\dcnotl{M_0}}}}}}
    }
  }
  \\
  &\dcveq
  \dccut{M_1}{
    \dcabsl{\dcvar''}{
      \dccut{\dcvar''}{\dcabsl{x}{\dccut{\dcnotr{K}}{\dcnotl{M_0}}}}
    }
  }
  \\
  &\dcveq
  \dccut{M_1}{
    \dcabsl{x}{\dccut{\dcnotr{K}}{\dcnotl{M_0}}}
  }
  \\
  &\dcveq
  \dccut{M_1}{
    \dcabsl{x}{\dccut{M_0}{K}}
  }
  \enspace.
  \end{align*}

  (b) If $M_1$ is a value (say $\dcValue$), then
  \begin{align*}
  & \dccut{(\lambda x.M_0)}{(\dcappl{\dcValue}{K})}\\
  &\equiv
  \dccut{
    \dcnotr{\dcabsl{\dcvar'}{\dccut{\dcvar'}{\dcfst{\dcabsl{x}{\dccut{\dcvar'}{\dcsnd{\dcnotl{M_0}}}}}}}}
  }{
    \dcnotl{\dcpairr{\dcValue}{\dcnotr{K}}}
  }
  \\
  &\dcveq
  \dccut{
    \dcpairr{\dcValue}{\dcnotr{K}}
  }{
    \dcabsl{\dcvar'}{\dccut{\dcvar'}{\dcfst{\dcabsl{x}{\dccut{\dcvar'}{\dcsnd{\dcnotl{M_0}}}}}}}
  }
  \\
  &\dcveq
  \dccut{
    \dcpairr{\dcValue}{\dcnotr{K}}
  }{
    \dcfst{\dcabsl{x}{\dccut{\dcpairr{\dcValue}{\dcnotr{K}}}{\dcsnd{\dcnotl{M_0}}}}}
  }
  \\
  &\dcveq
  \dccut{\dcValue}{
    \dcabsl{x}{
      \dccut{\dcnotr{K}}{\dcnotl{M_0}}
    }
  }
  \\
  &\dcveq
  \dccut{\dcValue}{
    \dcabsl{x}{
      \dccut{M_0}{K}
    }
  }
  \enspace.
  \end{align*}    
\end{proof}

  The but-not type $\dctype_0\dcpmi\dctype_1$ with its term $\dcappr{M}{K}$, coterm $\lambda\dccovar.K$ 
  are also defined as
  \begin{align*}
    \dctype_0\dcpmi\dctype_1
    &\equiv
    \dctype_0\land\negop{\dctype_1}\\
    \dcappr{\dcterm}{\dccoterm}
    &\equiv
    \dcpairr{\dcterm}{\dcnotr{\dccoterm}}
    &
    \lambda \alpha.K
    &\equiv
    \dcabsl{x}{\dccut{x}{\dcsnd{\dcnotl{\dcabsr{\alpha}{\dccut{x}{\dcfst{K}}}}}}}
    \enspace .
  \end{align*}

\begin{proposition}\label{prop:dcpmi}
  The following inferences are derivable:
  \begin{center}
    \AxiomC{
      $\varGamma \vdash \varDelta \mid M \colon \dctype_0$
    }
    \AxiomC{
      $K \colon \dctype_1 \mid \varGamma \vdash \varDelta$
    }    
    \BinaryInfC{
      $\varGamma \vdash \varDelta \mid \dcappr{\dcterm}{\dccoterm} \colon \dctype_0\dcpmi\dctype_1$
    }    
    \DisplayProof
    \qquad
    \AxiomC{
      $\dccoterm \colon \dctype_0 \mid \varGamma\vdash \varDelta, \dccovar \colon \dctype_1$
    }
    \UnaryInfC{
      $\lambda\dccovar.\dccoterm \colon \dctype_0\dcpmi\dctype_1 \mid \varGamma \vdash \varDelta$
    }
    \DisplayProof
    \enspace .
  \end{center}
  Also, $\dcappr{\dcValue}{K}$ is a value and the following equations hold:
  \begin{alignat*}{2}
    & (\beta\mathord\ot)\quad &
    \dccut{(\dcappr{\dcValue}{K_0})}{(\lambda \alpha.K_1)}
    &\dcveq
    \dccut{\dcabsr{\alpha}{\dccut{\dcValue}{K_1}}}{K_0}
    \\
    & (\eta\mathord\ot) &
    \dccoterm
    &\dcveq
    \lambda\dccovar.
    (\dcabsl{
      \dcvar
    }{
      \dccut{(\dcappl{\dccovar}{\dcvar})}{\dccoterm}
    })
    \\    
    & (\zeta\mathord\ot) &
    \dccut{(\dcappr{M}{K_0})}{K_1}
    &\dcveq
    \dccut{M}{\dcabsl{x}{\dccut{(\dcappr{x}{K_0})}{K_1}}}
    \tag*{if $x$ is fresh.}
  \end{alignat*}
\end{proposition}
\begin{proof}
  The inference part is shown immediately by the definition of $\dcappr{\dcterm}{\dccoterm}$ and $\lambda \alpha.K$.
  By the definition, it is immediately checked that a term of the form $\dcappr{\dcValue}{\dccoterm}$ is a value.
  
  The first equation $(\beta\mathord{\dcpmi})$ is shown as follows: 
  \begin{align*}
    \dccut{(\dcappr{\dcValue}{K_0})}{(\lambda \alpha.K_1)}
    &\equiv
    \dccut{
      \dcpairr{\dcValue}{\dcnotr{K_0}}
    }{
      \dcabsl{x}{\dccut{x}{\dcsnd{\dcnotl{\dcabsr{\alpha}{\dccut{x}{\dcfst{K_1}}}}}}}
    }
    \\
    &\dcveq
    \dccut{
      \dcpairr{\dcValue}{\dcnotr{K_0}}
    }{
      \dcsnd{\dcnotl{\dcabsr{\alpha}{
            \dccut{
              \dcpairr{\dcValue}{\dcnotr{K_0}}
            }{\dcfst{K_1}}}}}
    }
    \\
    &\dcveq
    \dccut{\dcnotr{K_0}}{
      \dcnotl{\dcabsr{\alpha}{\dccut{\dcValue}{K_1}}}
    }
    \\
    &\dcveq
    \dccut{
      \dcabsr{\alpha}{\dccut{\dcValue}{K_1}}
    }{
      K_0
    }
    \\
  \end{align*}
  
  The second equation $(\zeta\mathord{\dcpmi})$ is shown with $(\nu\wedge_0)$ as follows:
  \begin{align*}
    \dccut{(\dcappr{M}{K_0})}{K_1}
    &\equiv
    \dccut{
      \dcpairr{M}{\dcnotr{K_0}}
    }{K_1}
    \\
    &\dcveq
    \dccut{M}{
      \dcabsl{x}{
        \dccut{
          \dcpairr{x}{\dcnotr{K_0}}
        }{K_1}
      }      
    }
    \\
    &\equiv
    \dccut{M}{
      \dcabsl{x}{
        \dccut{
          (\dcappr{x}{K_0})
        }{K_1}
      }      
    }    
  \end{align*}
\end{proof}

We define a sub-calculus $\cbvSubDCeq$ of the call-by-value dual calculus that is obtained
by removing the negation type $\negop{A}$ and adding the implication and but-not types with their syntactical objects, typing rules, and equations
as primitives.
The calculus $\cbvSubDCeq$ can be roughly understood as a sub-calculus of the call-by-value dual calculus
that forbids free occurrences of the negation connective and allows only the occurrences necessary
to define the implication and the but-not connectives.

\subsection{Equivalence between $\cbvBLCeq$ and $\cbvSubDCeq$}

We define a translation $\fuga{-}$ from $\cbvBLCeq$ into
$\cbvSubDCeq$.  The translation is designed so that judgments in the
bilateral natural deduction are mapped to sequents in the sequent
calculus including proofs. 

We assume that there exist variables $\dccst{o}$ and covariables
$\dcbullet{o}$ of $\cbvSubDCeq$ for any constant expressions $\cst^o$
and constant continuations $\bullet^o$ of $\cbvBLCeq$, respectively.

Term $\fuga{\expr}$, coterm $\fuga{\cont}$, and statement $\fuga{\neutral}$
are defined inductively as shown in Figure~\ref{fig:BLC2DC}.
We show that the translation preserves typability.
Let ${\rm Cons}^+$ and ${\rm Cons}^-$ be the sets of constants of the form $\cst^o$ and $\bullet^o$, respectively.
For any $X \subseteq_{\rm {fin}} {\rm Cons}^+$ and $Y \subseteq_{\rm {fin}} {\rm Cons}^-$,
we respectively define
\begin{align*}
  \fuga{\varPi}_X = \varPi \cup \{x_{\cst^o} \colon o \mid \cst^o \in X\}
  \enspace \mbox{and} \enspace
  \fuga{\varSigma}_Y =
  \varSigma \cup \{\alpha_{\bullet^o} \colon o \mid \bullet^o \in Y\}
  \enspace .
\end{align*}

\begin{figure}[t]
\begin{center}
  \begin{minipage}{0.49\textwidth}
    \begin{align*}
      \fuga{ \cst^o } &\equiv \dccst{o}
      \\
      \fuga{ \evar^A } &\equiv \dcvar
      \\
      \fuga{ \lamabs{\evar^A}{\expr} } &\equiv \lamabs{\dcvar}{\fuga{\expr}}
      \\
      \fuga{ \app{\expr_0}{\expr_1} } &\equiv \dcabsr{\dccovar}{\dccut{\fuga{\expr_0}}{(\dcappl{\fuga{\expr_1}}{\dccovar})} }
      \\
      \fuga{ \pair{\expr_0}{\expr_1} } &\equiv \dcpairr{\fuga{\expr_0}}{\fuga{\expr_1}}
      \\
      \fuga{ \fst{\expr} } &\equiv \dcabsr{\dccovar}{\dccut{\fuga{\expr}}{\dcfst{\dccovar}}}
      \\
      \fuga{ \snd{\expr} } &\equiv \dcabsr{\dccovar}{\dccut{\fuga{\expr}}{\dcsnd{\dccovar}}}
      \\
      \fuga{ \muabs{\cvar^A}{\neutral} } &\equiv \dcabsr{\dccovar}{\fuga{\neutral}}
    \end{align*}
  \end{minipage}
  \begin{minipage}{0.49\textwidth}
    \begin{align*}
      \fuga{ \bullet^o } &\equiv \dcbullet{o}
      \\
      \fuga{ \cvar^A } &\equiv \dccovar
      \\
      \fuga{ \lamabs{\cvar^A}{\cont} } &\equiv \lamabs{\dccovar}{\fuga{\cont}}
      \\
      \fuga{ \app{\cont_0}{\cont_1} } &\equiv \dcabsl{\dcvar}{\dccut{(\dcappr{\dcvar}{\fuga{\cont_1}})}{\fuga{\cont_0}} }
      \\
      \fuga{ \pair{\cont_0}{\cont_1} } &\equiv \dcpairl{\fuga{\cont_0}}{\fuga{\cont_1}}
      \\
      \fuga{ \fst{\cont} } &\equiv \dcabsl{\dcvar}{\dccut{\dcinl{\dcvar}}{\fuga{\cont}}}
      \\
      \fuga{ \snd{\cont} } &\equiv \dcabsl{\dcvar}{\dccut{\dcinr{\dcvar}}{\fuga{\cont}}}
      \\
      \fuga{ \muabs{\evar^A}{\neutral} } &\equiv \dcabsl{\dcvar}{\fuga{\neutral}}
    \end{align*}
  \end{minipage}
  \\
\[
\fuga{\bracket{\expr}{\cont}} \equiv \dccut{\fuga{\expr}}{\fuga{\cont}} 
\]
\end{center}
\caption{A translation from $\cbvBLCeq$ into $\cbvSubDCeq$}\label{fig:BLC2DC}
\end{figure}

\begin{proposition}\label{prop:typing_trans}
  \begin{enumerate}
  \item
    $\plusj{\varPi; \varSigma}{\expr}{A}$
    implies
    $\fuga{\varPi}_X \vdash \fuga{\varSigma}_Y \mid \fuga{\expr}  \colon  A$
    for any ${\rm Cons}^+(\expr) \subseteq X$ and ${\rm Cons}^-(\expr) \subseteq Y$,
  \item
    $\minusj{\varPi;\varSigma}{\cont}{A}$
    implies
    $\fuga{\cont}  \colon  A \mid \fuga{\varPi}_X \vdash \fuga{\varSigma}_Y$
    for any ${\rm Cons}^+(\cont) \subseteq X$ and ${\rm Cons}^-(\cont) \subseteq Y$, and
  \item
    $\zeroj{\varPi;\varSigma}{\neutral}$
    implies
    $\fuga{\varPi}_X \mid \fuga{\neutral} \vdash \fuga{\varSigma}_Y$
    for any ${\rm Cons}^+(\neutral) \subseteq X$ and ${\rm Cons}^-(\neutral) \subseteq Y$.
  \end{enumerate}
\end{proposition}
\begin{proof}
  The claims are shown by simultaneous induction on the derivation of the bilateral $\lambda$-calculus.
\end{proof}

\begin{lemma}\label{lem:subst_trans}
  \begin{enumerate}
  \item $\expr$ is a value of $\cbvBLCeq$ if and only if $\fuga{\expr}$ is a value of $\cbvSubDCeq$ , and\label{lem:subst_trans_equiv}
  \item $\fuga{[\eValue/x]D} \equiv [\fuga{\eValue}/x]\fuga{D}$ and
    $\fuga{[\cont/\alpha]D} \equiv [\fuga{\cont}/\alpha]\fuga{D}$
    \enspace.\label{lem:subst_trans_subst}
  \end{enumerate}
\end{lemma}
\begin{proof}
  The claim (\ref{lem:subst_trans_equiv}) can be shown immediately.

  The former claim of (2) is shown by induction on $S$. 
  Note that, by (1),
  $\nyo{\expr}$ is a value if and only if $\nyo{[\eValue/x]\expr}$ is a value. 
  
  The case of $S \equiv \cst^o$:
  \begin{align*}
    \nyo{[\eValue/x]\cst^o}
    \equiv \nyo{\cst^o}
    \equiv x_{\cst^o}
    \equiv [\nyo{\eValue}/x]x_{\cst^o}
    \equiv [\nyo{\eValue}/x]\nyo{\cst^o}
  \end{align*}    
  
  The case of $S \equiv x^A$:
  \begin{align*}
    \nyo{[\eValue/x]x^A}
    \equiv \nyo{\eValue}
    \equiv [\nyo{\eValue}/x]x
    \equiv [\nyo{\eValue}/x]\nyo{x^A}
  \end{align*}

  The case of $S \equiv x_0^{A_0}$, where $x^A \not\equiv x_0^{A_0}$:
  \begin{align*}
    \nyo{[\eValue/x]x_0^{A_0}}
    \equiv \nyo{x_0^{A_0}}
    \equiv x_0
    \equiv [\nyo{\eValue}/x]x_0
    \equiv [\nyo{\eValue}/x]\nyo{x_0^{A_0}}
  \end{align*}

  The case of $S \equiv \fst{\expr}$ is shown by using induction hypothesis. 
  \begin{align*}
    \nyo{[\eValue/x]\fst{\expr}}
    &\equiv \nyo{\fst{[\eValue/x]\expr}}
    \\
    &\equiv
    \dcabsr{\overline{\alpha}}{
      \dccut{
        \nyo{[\eValue/x]\expr}
      }{
        \dcabsl{\overline{x'}}{
          \dccut{x'}{\dcfst{\alpha}}
        }
      }
    }
    \\
    &\equiv
    \dcabsr{\overline{\alpha}}{
      \dccut{
        [\nyo{\eValue}/x]\nyo{\expr}
      }{
        \dcabsl{\overline{x'}}{
          \dccut{x'}{\dcfst{\alpha}}
        }
      }
    }
    \tag*{by induction hypothesis}
    \\
    &\equiv
    [\nyo{\eValue}/x]
    (\dcabsr{\overline{\alpha}}{
      \dccut{
        \nyo{\expr}
      }{
        \dcabsl{\overline{x'}}{
          \dccut{x'}{\dcfst{\alpha}}
        }
      }
    })
    \\
    &\equiv
    [\nyo{\eValue}/x](\fst{\expr})
  \end{align*}

  The other cases
  $S \equiv \snd{\expr}$,
  $\pair{\expr_0}{\expr_1}$,
  $\lamabs{\evar}{\expr}$,
  $\app{\expr_0}{\expr_1}$,
  $\muabs{\cvar}{\neutral}$,
  $\bullet^o$,
  $\cvar^A$,
  $\fst{\cont}$,
  $\snd{\cont}$,
  $\pair{\cont_0}{\cont_1}$,
  $\lamabs{\cvar}{\cont}$, 
  $\app{\cont_0}{\cont_1}$, 
  $\muabs{\evar}{\neutral}$, and 
  $\bracket{\expr}{\cont}$
  are also shown straightforwardly by using induction hypothesis. 
  
  The latter claim of (2) is shown by induction on $S$.

  The case of $S \equiv \bullet^o$:
  \begin{align*}
    \nyo{[\cont/\alpha]\bullet^o}
    \equiv
    \nyo{\bullet^o}
    \equiv 
    \alpha_{\bullet^o}
    \equiv
    [\nyo{\cont}/\alpha]\alpha_{\bullet^o}
    \equiv
    [\nyo{\cont}/\alpha]\nyo{\bullet^o}
  \end{align*}
  
  The case of $S \equiv \cvar^A$:
  \begin{align*}
    \nyo{[\cont/\alpha]\cvar^A}
    \equiv
    \nyo{\cont}
    \equiv
    [\nyo{\cont}/\alpha]\alpha
    \equiv
    [\nyo{\cont}/\alpha]\nyo{\cvar^A}
  \end{align*}
  
  The case of $S \equiv \cvar_0^{A_0}$, where $\cvar^A \not\equiv \cvar_0^{A_0}$:
  \begin{align*}
    \nyo{[\cont/\alpha]\cvar_0^{A_0}}
    \equiv
    \nyo{\cvar_0^{A_0}}
    \equiv
    \cvar_0
    \equiv
    [\nyo{\cont}/\alpha]\cvar_0
    \equiv
    [\nyo{\cont}/\alpha]\nyo{\cvar_0^{A_0}}
  \end{align*}

  The case of $S \equiv \fst{\cont_0}$ is shown by using induction hypothesis:
  \begin{align*}
    \nyo{[\cont/\alpha]\fst{\cont_0}}
    &\equiv
    \nyo{\fst{[\cont/\alpha]\cont_0}}
    \\
    &\equiv
    \dcabsl{x}{\dccut{\dcinl{x}}{\nyo{[\cont/\alpha]\cont_0}}}
    \\
    &\equiv
    \dcabsl{x}{\dccut{\dcinl{x}}{[\nyo{\cont}/\alpha]\nyo{\cont_0}}}
    \tag*{by induction hypothesis}
    \\
    &\equiv
    [\nyo{\cont}/\alpha](\dcabsl{x}{\dccut{\dcinl{x}}{\nyo{\cont_0}}})
    \\
    &\equiv
    [\nyo{\cont}/\alpha]\nyo{\fst{\cont_0}}
  \end{align*}

  The other cases
  $S \equiv \cst^o$, 
  $x^A$,
  $\fst{\expr}$,
  $\snd{\expr}$,
  $\pair{\expr_0}{\expr_1}$,
  $\lamabs{\evar}{\expr}$,
  $\app{\expr_0}{\expr_1}$,
  $\muabs{\cvar}{\neutral}$,
  $\snd{\cont}$,
  $\pair{\cont_0}{\cont_1}$,
  $\lamabs{\cvar}{\cont}$, 
  $\app{\cont_0}{\cont_1}$, 
  $\muabs{\evar}{\neutral}$, and 
  $\bracket{\expr}{\cont}$
  are also shown straightforwardly by using induction hypothesis.
\end{proof}

\begin{theorem}\label{thm:BLC2DC}
    $D_0 \cbveq D_1$
    implies
    $\fuga{D_0} \dcveq \fuga{D_1}$.
\end{theorem}
\begin{proof}
First, we define a coterm $\dccoterm$-indexed translation
$\fuga{-}_\dccoterm$ from contexts for expressions of $\cbvBLCeq$ into
contexts of $\cbvSubDCeq$ as follows:
\begin{align*}
  \fuga{\{-\}}_\dccoterm & \equiv \dccut{\{-\}}{\dccoterm}\\
  \fuga{\eValue\eEvalSymb}_\dccoterm & \equiv \fuga{\eEvalSymb}_{\dcabsl{\dcvar}{\dccut{\fuga{\eValue}}{(\dcappl{\dcvar}{\dccoterm})}}} &
  \fuga{\eEvalSymb\expr}_\dccoterm & \equiv \fuga{\eEvalSymb}_{\dcappl{\fuga{\expr}}{\dccoterm}} &
  \\
  \fuga{\pair{\eValue}{\eEvalSymb}}_\dccoterm & \equiv \fuga{\eEvalSymb}_{\dcabsl{\dcvar}{\dccut{\dcpairr{\fuga{\eValue}}{\dcvar}}{\dccoterm}}} &
  \fuga{\pair{\eEvalSymb}{\expr}}_\dccoterm & \equiv \fuga{\eEvalSymb}_{\dcabsl{\dcvar}{\dccut{\dcpairr{\dcvar}{\fuga{\expr}}}{\dccoterm}}} &
  \\
  \fuga{\fst{\eEvalSymb}}_\dccoterm & \equiv \fuga{\eEvalSymb}_{\dcfst{\dccoterm}} &
  \fuga{\snd{\eEvalSymb}}_\dccoterm & \equiv \fuga{\eEvalSymb}_{\dcsnd{\dccoterm}}
  \enspace .
\end{align*}

Next, we can immediately confirm
\begin{align*}
\dccut{\fuga{\eEval{\expr}}}{\dccoterm} \dcveq \dccut{\fuga{\expr}}{\dcabsl{\dcvar}{\fuga{\eEvalSymb}_{\dccoterm}\{\dcvar\}}} \tag{$\ast$}
\end{align*}
 by induction on $\eEvalSymb$.

 Finally, the theorem can be shown by the case analysis of $\cbveq$ and Lemma~\ref{lem:subst_trans}.

  The case of $\app{(\lamabs{\evar}{\expr})}{\eValue} \cbveq [\eValue/\evar] \expr$.
  By using Lemma~\ref{lem:subst_trans} (1) and (2), we have 
  \begin{align*}
  \fuga{\app{(\lamabs{\evar}{\expr})}{\eValue}}
  &\equiv
  \dcabsr
      {
        \dccovar
      }{
        \dccut{
          \lambda x.\fuga{\expr}
        }{
          (\dcappl{\fuga{\eValue}}{\dccovar})
        }
      }
  \dcveq
  \dcabsr{
    \dccovar
  }{
    \dccut{
      \fuga{\eValue}
    }{
      \dcabsl{x}{
        \dccut{\fuga{\expr}}{\dccovar}
      }
    }
  }
  \\
  &\dcveq
  \dcabsr{
    \dccovar
  }{
    \dccut{[\fuga{\eValue}/\evar]\fuga{\expr}}{\dccovar}
  }
  \equiv
  \dcabsr{
    \dccovar
  }{
    \dccut{\fuga{[\eValue / \evar] \expr}}{\dccovar}
  }
  \dcveq
  \fuga{[\eValue / \evar] \expr}\enspace. 
  \end{align*}
  
  The case of $\lamabs{\evar}{\app{\eValue}{\evar}} \cbveq \eValue$. 
  By using Lemma~\ref{lem:subst_trans} (1), we have
  \begin{align*}
  \fuga{\lamabs{\evar}{\app{\eValue}{\evar}}}
  &\equiv
  \lamabs{\dcvar}
         {(
           \dcabsr{
             \dccovar
           }{
             \dccut{
               \fuga{\eValue}
             }{
               (\dcappl{
                 \dcvar
               }{
                 \dccovar
               })
             }
           }
         )}
  \dcveq
  \fuga{\eValue}
  \enspace.
  \end{align*}
  
  The case of $\fst{\pair{\eValue_0}{\eValue_1}} \cbveq \eValue_0$. 
  By using Lemma~\ref{lem:subst_trans} (1), we have
  \begin{eqnarray*}
    \fuga{
      \fst{\pair{\eValue_0}{\eValue_1}}
    }
    \equiv
    \dcabsr{
      \dccovar
    }{
      \dccut{
        \dcpairr{\fuga{\eValue_0}}{\fuga{\eValue_1}}
      }{
        \dcfst{\dccovar}
      }
    }
    \dcveq
    \dcabsr{\dccovar}{
      \dccut{
        \fuga{\eValue_0}
      }{
        \dccovar
      }
    }
    \dcveq
    \fuga{\eValue_0}
    \enspace.
  \end{eqnarray*}
  
  The case of $\snd{\pair{\eValue_0}{\eValue_1}} \cbveq \eValue_1$ is shown similarly. 

  The case of $\pair{\fst{\eValue}}{\snd{\eValue}} \cbveq \eValue$.
  By using Lemma~\ref{lem:subst_trans} (1), we have
  \begin{eqnarray*}
    \fuga{
      \pair{\fst{\eValue}}{\snd{\eValue}}
    }
    \equiv
    \dcpairr{
      \dcabsr{
        \dccovar
      }{
        \dccut{
          \fuga{\eValue}
        }{
          \dcfst{\dccovar}
        }
      }
    }{
      \dcabsr{
        \dccovar
      }{
        \dccut{
          \fuga{\eValue}
        }{
          \dcsnd{\dccovar}
        }
      }
    }
    \dcveq
    \fuga{\eValue}
    \enspace.
  \end{eqnarray*}  

  The case of $\muabs{\cvar}{\bracket{\expr}{\cvar}} \cbveq \expr$.
  We have 
  \begin{eqnarray*}
    \fuga{
      \muabs{\cvar}{\bracket{\expr}{\cvar}}
    }
    \equiv
    \dcabsr{
      \dccovar
    }{
      \dccut{
        \fuga{\expr}
      }{
        \dccovar
      }
    }
    \dcveq
    \fuga{\expr}
    \enspace.
  \end{eqnarray*}  

  The case of $\bracket{\eValue}{\muabs{\evar}{\neutral}} \cbveq [\eValue/\evar] N$:
  By using Lemma~\ref{lem:subst_trans} (1) and (2), we have
  \begin{align*}
    \fuga{
      \bracket{\eValue}{\muabs{\evar}{\neutral}}
    }
    \equiv
    \dccut{
      \fuga{\eValue}
    }{
      \dcabsl{
        \dcvar
      }{
        \fuga{\neutral}
      }
    }
    \dcveq
    [\fuga{\eValue}/\dcvar] \fuga{\neutral}
    \equiv
    \fuga{[\eValue/\dcvar]\neutral}
    \enspace.
  \end{align*}

  The case of $\bracket{\eEval{\expr}}{\cont} \cbveq \bracket{\expr}{\muabs{\evar}{\bracket{\eEval{\evar}}{\cont}}}$. 
  By using $(\ast)$. we have
  \begin{align*}
    \fuga{
      \bracket{\eEval{\expr}}{\cont}
    }
    &\equiv
    \dccut{
      \fuga{\eEval{\expr}}
    }{
      \fuga{\cont}
    }
    \dcveq
    \dccut{
      \fuga{\expr}
    }{
      \dcabsl{
        \dcvar
      }{
        \fuga{\eEvalSymb}_{\fuga{\cont}}\{\dcvar\}
      }
    }
    \\
    &\dcveq
    \dccut{
      \fuga{\expr}
    }{
      \dcabsl{
        \dcvar
      }{
        \dccut{
          \fuga{\dcvar}
        }{
          \dcabsl{
            \dcvar
          }{
            \fuga{\eEvalSymb}_{\fuga{\cont}}\{\dcvar\}
          }
        }
      }
    }
    \\
    &\dcveq
    \dccut{
      \fuga{\expr}
    }{
      \dcabsl{
        \dcvar
      }{
        \dccut{
          \fuga{\eEval{\evar}}
        }{
          \fuga{\cont}
        }
      }
    }
    \equiv
    \fuga{
      \bracket{
        \expr
      }{
        \muabs{\evar}{\bracket{\eEval{\evar}}{\cont}}
      }
    }
    \enspace.
  \end{align*}

  The case of $\app{(\lamabs{\cvar}{\cont_0})}{\cont_1} \cbveq [\cont_1/\cvar] \cont_0$. 
  By using Lemma~\ref{lem:subst_trans} (2), we have
  \begin{align*}
    \fuga{
      \app{(\lamabs{\cvar}{\cont_0})}{\cont_1}
    }
    &\equiv
    \dcabsl{
      \dcvar
    }{
      \dccut{
        (\dcappr{\fuga{\cont_1}}{\dcvar})
      }{
        \lamabs{\dccovar}{\fuga{\cont_0}}
      }
    }
    \dcveq
    \dcabsl{
      \dcvar
    }{
      \dccut{
        \dcabsr{\dccovar}{
          \dccut{
            \dcvar
          }{
            \fuga{\cont_0}
          }
        }
      }{
        \fuga{\cont_1}
      }
    }
    \\
    &\dcveq
    \dcabsl{
      \dcvar
    }{
      \dccut{
        \dcvar
      }{
        [\fuga{\cont_1}/\dccovar]\fuga{\cont_0}
      }
    }
    \dcveq
    [\fuga{\cont_1}/\dccovar]\fuga{\cont_0}
    \equiv
    \fuga{ [\cont_1/\dccovar]\cont_0 }
    \enspace.
  \end{align*}

  The case of $\lamabs{\cvar}{\app{\cont}{\cvar}} \cbveq \cont$.
  We have
  \begin{align*}
    \fuga{
      \lamabs{\cvar}{\app{\cont}{\cvar}}
    }
    &\equiv
    \lamabs{\dccovar}{
      (\dcabsl{\dcvar}{
        \dccut{
          (\dcappr{\dccovar}{\dcvar})
        }{
          \fuga{\cont}
        }
      })
    }
    \dcveq
    \fuga{\cont}
    \enspace.
  \end{align*}

  The case of $\fst{\pair{\cont_0}{\cont_1}} \cbveq \cont_0$. 
  We have
  \begin{align*}
    \fuga{
      \fst{\pair{\cont_0}{\cont_1}}
    }
    &\equiv
    \dcabsl{
      \dcvar
    }{
      \dccut{
        \dcinl{\dcvar}
      }{
        \dcpairl{\fuga{\cont_0}}{\fuga{\cont_1}}
      }
    }
    \dcveq
    \dcabsl{
      \dcvar
    }{
      \dccut{
        \dcvar
      }{
        \fuga{\cont_0}
      }
    }
    \dcveq    
    \fuga{\cont_0}
    \enspace.
  \end{align*}

  The case of $\snd{\pair{\cont_0}{\cont_1}} \cbveq \cont_1$ is also shown similarly. 

  The case of $\pair{\fst{\cont}}{\snd{\cont}} \cbveq \cont$. 
  We have
  \begin{align*}
    \fuga{
      \pair{\fst{\cont}}{\snd{\cont}}
    }
    &\equiv
    \dcpairl{
      \dcabsl{
        \dcvar
      }{
        \dccut{
          \dcinl{\dcvar}
        }{
          \fuga{\cont}
        }
      }
    }{
      \dcabsl{
        \dcvar
      }{
        \dccut{
          \dcinr{\dcvar}
        }{
          \fuga{\cont}
        }
      }      
    }
    \dcveq
    \fuga{\cont}
    \enspace.
  \end{align*}

  The case of $\muabs{\evar}{\bracket{\evar}{\cont}} \cbveq \cont$. 
  We have
  \begin{align*}
    \fuga{
      \muabs{\evar}{\bracket{\evar}{\cont}}
    }
    &\equiv
    \dcabsl{
        \dcvar
    }{
      \dccut{
        \dcvar
      }{
        \fuga{\cont}
      }
    }
    \dcveq
    \fuga{\cont}
    \enspace.
  \end{align*}

  
  The case of $\bracket{\muabs{\cvar}{\neutral}}{\cont} \cbveq [\cont/\cvar]$. 
  By using Lemma~\ref{lem:subst_trans} (2), we have
  \begin{align*}
    \fuga{
      \bracket{\muabs{\cvar}{\neutral}}{\cont}
    }
    &\equiv
    \dccut{
      \dcabsr{
        \dccovar
      }{
        \fuga{\neutral}
      }
    }{
     \fuga{\cont} 
    }
    \dcveq
    [\fuga{\cont}/\dccovar]\fuga{\neutral}
    \equiv
    \fuga{ [\cont/\dccovar]\neutral }
    \enspace.
  \end{align*}
\end{proof}

We next define a translation $\gafu{-}$
from $\cbvSubDCeq$ into $\cbvBLCeq$.
Expression $\gafu{\dcterm}$, continuation $\gafu{\dccoterm}$, and command $\gafu{\dcstat}$
for any typable $\dcterm$, $\dccoterm$, and $\dcstat$ in $\cbvSubDCeq$ are defined inductively as shown in Figure~\ref{fig:DC2BLC}.

We show that this translation preserves typability.
Let $\varGamma$ be
$\vec{\dccst{o'}\colon o'},\vec{\dcvar\colon A}$.
Then we define $\gafu{\varGamma}$ by $\vec{x\colon\plusop{A}}$.
Similarly, we also define $\gafu{\varDelta}$.

\begin{figure}[tp]
\begin{center}
  \begin{minipage}{0.49\textwidth}
    \begin{align*}
      \gafu{ \dccst{o} } &\equiv \cst^o
      \\
      \gafu{ \dcvar } &\equiv \evar^A
      \tag*{if $\dcvar$ has a type $A$}
      \\
      \gafu{ \lamabs{\dcvar}{\dcterm} } &\equiv \lamabs{\evar^A}{\gafu{\dcterm}}
      \tag*{if $\lamabs{\dcvar}{\dcterm}$ has a type $A\to A_1$}
      \\
      \gafu{ \dcappr{\dcterm}{\dccoterm} } &\equiv \muabs{\cvar}{\bracket{\gafu{\dcterm}}{\app{\cvar}{\gafu{\dccoterm}}}}
      \\
      \gafu{ \dcpairr{\expr_0}{\expr_1} } &\equiv \pair{\gafu{\expr_0}}{\gafu{\expr_1}}
      \\
      \gafu{ \dcinl{\dcterm} } &\equiv \muabs{\cvar}{\bracket{\gafu{\dcterm}}{\fst{\cvar}}}
      \\
      \gafu{ \dcinr{\dcterm} } &\equiv \muabs{\cvar}{\bracket{\gafu{\dcterm}}{\snd{\cvar}}}
      \\
      \gafu{ \dcabsr{\dccovar}{\dcstat} } &\equiv \muabs{\cvar^A}{\gafu{\dcstat}}
      \quad \hbox{if $\dccovar\colon A$}
    \end{align*}
  \end{minipage}
  \begin{minipage}{0.49\textwidth}
    \begin{align*}
      \gafu{ \dcbullet{o} } &\equiv \bullet^o
      \\
      \gafu{ \dccovar } &\equiv \cvar^A
      \tag*{if $\dccovar$ has a type $A$}
      \\
      \gafu{ \lamabs{\dccovar}{\dccoterm} } &\equiv \lamabs{\cvar^A}{\gafu{\dccoterm}}
      \tag*{if $\lamabs{\dccovar}{\dccoterm}$ has a type $A_1 \leftarrow A$}
      \\
      \gafu{ \dcappl{\dcterm}{\dccoterm} } &\equiv \muabs{\evar}{\bracket{\app{\evar}{\gafu{\dcterm}}}{\gafu{\dccoterm}} }
      \\
      \gafu{ \dcpairl{\dccoterm_0}{\dccoterm_1} } &\equiv \pair{\gafu{\dccoterm_0}}{\gafu{\dccoterm_1}}
      \\
      \gafu{ \dcfst{\dccoterm} } &\equiv \dcabsl{\evar}{\bracket{\fst{\evar}}{\gafu{\dccoterm}}}
      \\
      \gafu{ \dcsnd{\dccoterm} } &\equiv \dcabsl{\evar}{\bracket{\snd{\evar}}{\gafu{\dccoterm}}}
      \\
      \gafu{ \dcabsl{\dcvar}{\dcstat} } &\equiv \muabs{\evar^A}{\gafu{\dcstat}}
      \quad \hbox{if $\dcvar\colon A$}
    \end{align*}
  \end{minipage}
\end{center}
  \[
  \gafu{\dccut{\dcterm}{\dccoterm}} \equiv \bracket{\gafu{\dcterm}}{\gafu{\dccoterm}}
  \]
\caption{A translation from $\cbvSubDCeq$ into $\cbvBLCeq$.}\label{fig:DC2BLC}
\end{figure}

\begin{proposition}
  \begin{enumerate}
  \item
    $\varGamma \vdash \varDelta \mid \dcterm\colon A$
    implies
    $\plusj{\gafu{\varGamma};\gafu{\varDelta}}{\gafu{\dcterm}}{A}$,
  \item
    $\dccoterm\colon A \mid \varGamma \vdash \varDelta$
    implies
    $\minusj{\gafu{\varGamma};\gafu{\varDelta}}{\gafu{\dccoterm}}{A}$, and
  \item
    $\varGamma \mid \dcstat \vdash \varDelta$
    implies
    $\zeroj{\gafu{\varGamma};\gafu{\varDelta}}{\gafu{\dcstat}}$. 
  \end{enumerate}
\end{proposition}
\begin{proof}
  The claims can be shown by simultaneous induction on the derivation of judgments of $\cbvSubDCeq$.
\end{proof}

We show that the translation $\gafu{-}$ is an inverse of $\fuga{-}$ up to $\cbveq$ as follows:
\begin{theorem}\label{lem:involution_trans_inv}
  \begin{enumerate}
  \item $\gafu{\fuga{D}} \cbveq D$ holds, and\label{lem:involution_trans_inv_BLC}
  \item $\fuga{\gafu{O}} \dcveq O$ holds.\label{lem:involution_trans_inv_DC}
  \end{enumerate}
\end{theorem}
\begin{proof}
(\ref{lem:involution_trans_inv_BLC}) is shown by induction on $D$. 
(\ref{lem:involution_trans_inv_DC}) is shown by induction on $O$. 
\end{proof}

\begin{lemma}\label{lem:subst_trans_inv}
  \begin{enumerate}
  \item
    $\dcterm$ is a value of $\cbvSubDCeq$ if and only if
    there exists a value $\eValue$ of $\cbvBLCeq$ such that $\eValue \cbveq \gafu{\dcterm}$, and\label{lem:subst_trans_inv_surj}
  \item
    $\gafu{[\dcValue/\dcvar]O} \equiv [\gafu{\dcValue}/\evar]\gafu{O}$ and
    $\gafu{[\dccoterm/\dccovar]O} \equiv [\gafu{\dccoterm}/\cvar]\gafu{O}$.\label{lem:subst_trans_inv_subst}
  \end{enumerate}
\end{lemma}
\begin{proof}
  The claim (\ref{lem:subst_trans_inv_surj}) is shown immediately.
  The claims of (\ref{lem:subst_trans_inv_subst}) are shown by induction on $O$. 
\end{proof}

For any $\dcEvalSymb$ (contexts of $\cbvSubDCeq$),
we define $\gafu{\dcEvalSymb}$ (contexts for expressions of $\cbvBLCeq$)
as follows:
  \begin{align*}
    \gafu{\{-\}} &\equiv \{-\}&
    \gafu{\dcappr{\dcEvalSymb}{\dccoterm}} &\equiv \muabs{\cvar}{\bracket{\gafu{\dcEvalSymb}}{\app{\cvar}{\gafu{\dccoterm}}}}
    \\
    \gafu{\dcpairr{\dcValue}{\dcEvalSymb}} &\equiv \pair{\gafu{\dcValue}}{\gafu{\dcEvalSymb}} &
    \gafu{\dcpairr{\dcEvalSymb}{\dcterm}} &\equiv \pair{\gafu{\dcEvalSymb}}{\gafu{\dcterm}}\\
    \gafu{\dcinl{\dcEvalSymb}} &\equiv \muabs{\cvar}{\bracket{\gafu{\dcEvalSymb}}{\fst{\cvar}}} &
    \gafu{\dcinr{\dcEvalSymb}} &\equiv \muabs{\cvar}{\bracket{\gafu{\dcEvalSymb}}{\snd{\cvar}}}
    \enspace .
  \end{align*}

\begin{lemma}\label{lem:evalcontext_trans_inv}
  \begin{enumerate}
  \item $\gafu{\dcEval{\dcterm}} \cbveq \gafu{\dcEvalSymb}\{\gafu{\dcterm}\}$ holds, and
  \item $\bracket{\gafu{\dcEvalSymb}\{\expr\}}{\cont} \cbveq
    \bracket{\expr}{\muabs{\evar}{\bracket{\gafu{\dcEvalSymb}\{\evar\}}{\cont}}}$ holds.
  \end{enumerate}
\end{lemma}
\begin{proof}
  The claims (1) and (2) are shown by induction on $\dcEvalSymb$.
\end{proof}

\begin{theorem}\label{thm:DC2BLC}
    $O_0 \dcveq O_1$
    implies
    $\gafu{O_0} \cbveq \gafu{O_1}$.
\end{theorem}
\begin{proof}
  The claim is shown by the case analysis of $\dcveq$. 

  The case of $(\beta\mathord\to)$ is shown as follows. 
  \begin{align*}
    \gafu{
      \dccut{
        (\lamabs{\dcvar}{\dcterm_0})
      }{
        (\dcappl{\dcterm_1}{\dccoterm})
      }
    }
    &\equiv
    \bracket{
      \lamabs{\evar}{\gafu{\dcterm_0}}
    }{
      \muabs{\evar_1}{
        \bracket{
          \app{\evar_1}{\gafu{\dcterm_1}}
        }{
          \gafu{\dccoterm}
        }
      }
    }
    \\
    &\cbveq
    \bracket{
      \app{(\lamabs{\evar}{\gafu{\dcterm_0}})}{\gafu{\dcterm_1}}
    }{
      \gafu{\dccoterm}
    }
    \\
    &\cbveq
    \bracket{
      \gafu{\dcterm_1}
    }{
      \muabs{\evar}{
        \bracket{
          \app{(\lamabs{\evar}{\gafu{\dcterm_0}})}{\evar}
        }{
          \gafu{\dccoterm}
        }
      }
    }
    \\
    &\cbveq
    \bracket{
      \gafu{\dcterm_1}
    }{
      \muabs{\evar}{
        \bracket{
          \gafu{\dcterm_0}
        }{
          \gafu{\dccoterm}
        }
      }
    }
    \\
    &\equiv
    \gafu{
      \dccut{
        \dcterm_1
      }{
        \dcabsl{\dcvar}{
          \dccut{
            \dcterm_0
          }{
            \dccoterm
          }
        }
      }
    }
    \enspace.
  \end{align*}
  
  The case of $(\beta\mathord\ot)$ is shown as follows.
  \begin{align*}
    \gafu{
      \dccut{
        (\dcappr{\dcterm}{\dccoterm_0})
      }{
        (\lamabs{\dccovar}{\dccoterm_1})
      }
    }
    &\equiv
    \bracket{
      \muabs{\cvar_1}{
        \bracket{
          \gafu{\dcterm}
        }{
          \app{\cvar_1}{\gafu{\dccoterm_0}}
        }
      }
    }{
      \lamabs{\cvar}{\gafu{\dccoterm_1}}
    }
    \\
    &\cbveq
    \bracket{
      \gafu{\dcterm}
    }{
      \app{(\lamabs{\cvar}{\gafu{\dccoterm_1}})}{\gafu{\dccoterm_0}}
    }
    \\
    &\cbveq
    \bracket{
      \gafu{\dcterm}
    }{
      [\gafu{\dccoterm_0}/\cvar]\gafu{\dccoterm_1}
    }
    \\
    &\equiv
    [\gafu{\dccoterm_0}/\cvar]
    \bracket{
      \gafu{\dcterm}
    }{
      \gafu{\dccoterm_1}
    }
    \\
    &\cbveq
    \bracket{
      \muabs{\cvar}{
        \bracket{
          \gafu{\dcterm}
        }{
          \gafu{\dccoterm_1}
        }
      }
    }{
      \gafu{\dccoterm_0}
    }
    \\    
    &\equiv
    \gafu{
      \dccut{
        \dcabsr{\dccovar}{
          \dccut{
            \dcterm
          }{
            \dccoterm_1
          }
        }
      }{
        \dccoterm_0
      }
    }
    \enspace.
  \end{align*}
  
  The case of $(\beta\land_0)$. 
  By using Lemma~\ref{lem:subst_trans_inv} (1), We have 
  \begin{align*}
    \gafu{\dccut{\dcpairr{\dcValue_0}{\dcValue_1}}{\dcfst{\dccoterm}}}
    &\equiv
    \bracket{
      \pair{\gafu{\dcValue_0}}{\gafu{\dcValue_1}}
      }{
      \muabs{\evar}{
        \bracket{
          \fst{\evar}
        }{
          \gafu{\dccoterm}
        }
      }
    }
    \\
    &\cbveq
    \bracket{
      \fst{\pair{\gafu{\dcValue_0}}{\gafu{\dcValue_1}}}
    }{
      \gafu{\dccoterm}
    }
    \\
    &\cbveq
    \bracket{
      \gafu{\dcValue_0}
    }{
      \gafu{\dccoterm}
    }
    \equiv
    \gafu{\dccut{\dcValue_0}{\dccoterm}}
    \enspace.    
  \end{align*}

  The case of $(\beta\land_1)$ is shown similarly. 

  The case of $(\beta\vee_0)$. 
  We have
  \begin{align*}
    \gafu{\dccut{\dcinl{\dcValue}}{\dcpairl{\dccoterm_0}{\dccoterm_1}}}
    &\equiv
    \bracket{
      \muabs{\cvar}{
        \bracket{\gafu{\dcValue}}{\fst{\cvar}}
      }
    }{
      \pair{
        \gafu{\dccoterm_0}
      }{
        \gafu{\dccoterm_1}
      }
    }
    \\
    &\cbveq
    \bracket{\gafu{\dcValue}}{
      \fst{
        \pair{
          \gafu{\dccoterm_0}
        }{
          \gafu{\dccoterm_1}
        }
      }
    }
    \\
    &\cbveq
    \bracket{
      \gafu{\dcValue}
    }{
      \gafu{\dccoterm_0}
    }
    \equiv
    \gafu{\dccut{\dcValue}{\dccoterm_0}}
    \enspace.    
  \end{align*}

  The case of $(\beta\vee_1)$ is shown similarly.

  The case of $(\beta R)$.
  By using Lemma~\ref{lem:subst_trans_inv} (1) and (2), We have
  \begin{align*}
    \gafu{\dccut{\dcValue}{\dcabsl{\dcvar}{\dcstat}}}
    &\equiv
    \bracket{
      \gafu{\dcValue}
    }{
      \muabs{\evar}{\gafu{\dcstat}}
    }
    \cbveq
    [\gafu{\dcValue}/\evar]\gafu{\dcstat}
    \equiv
    \gafu{[\dcValue/\dcvar]\dcstat}
    \enspace.    
  \end{align*}

  The case of $(\beta L)$.
  By using Lemma~\ref{lem:subst_trans_inv} (2), We have
  \begin{align*}
    \gafu{\dccut{\dcabsr{\dccovar}{\dcstat}}{\dccoterm}}
    &\equiv
    \bracket{
      \muabs{\cvar}{\gafu{\dcstat}}
    }{
      \gafu{\dccoterm}
    }
    \cbveq
    [\gafu{\dccoterm}/\cvar]\gafu{\dcstat}
    \equiv
    \gafu{[\dccoterm/\dccovar]\dcstat}
    \enspace.
  \end{align*}  

  The case of $(\eta\mathord\to)$ is shown by using Lemma~\ref{lem:subst_trans_inv} (1). 
  \begin{align*}
    \gafu{
      \lamabs{\dcvar}{
        (\dcabsr{
          \dccovar
        }{
          \dccut{\dcValue}{(\dcappl{\dcvar}{\dccovar})}
        })
      }
    }
    &\equiv
    \lamabs{\evar}{
      \muabs{\cvar}{
        \bracket{
          \gafu{\dcValue}
        }{
          \muabs{\evar_1}{
            \bracket{
              \app{\evar_1}{\evar}
            }{
              \cvar
            }
          }
        }
      }
    }
    \\
    &\cbveq
    \lamabs{\evar}{
      \muabs{\cvar}{
        \bracket{
          \app{\gafu{\dcValue}}{\evar}
        }{
          \cvar
        }
      }
    }
    \cbveq
    \lamabs{\evar}{\app{\gafu{\dcValue}}{\evar}}
    \cbveq
    \gafu{\dcValue}
    \enspace.    
  \end{align*}


  The case of $(\eta\mathord\ot)$.
  \begin{align*}
    \gafu{
      \lamabs{\dccovar}{
        (\dcabsl{
          \dcvar
        }{
          \dccut{(\dcappr{\dcvar}{\dccovar})}{\dccoterm}
        })
      }
    }
    &\equiv
    \lamabs{\cvar}{
      \muabs{\evar}{
        \bracket{
          \muabs{\cvar_1}{
            \bracket{
              \app{\cvar_1}{\cvar}
            }{
              \evar
            }
          }
        }{
          \gafu{\dccoterm}
        }
      }
    }
    \\
    &\cbveq
    \lamabs{\cvar}{
      \muabs{\evar}{
        \bracket{
          \app{\gafu{\dccoterm}}{\cvar}
        }{
          \evar
        }
      }
    }
    \cbveq
    \lamabs{\cvar}{\app{\gafu{\dccoterm}}{\cvar}}
    \cbveq
    \gafu{\dccoterm}
    \enspace.
  \end{align*}
  
  
  The case of $(\eta\land)$.
  Note that
  $\gafu{\dcabsr{\dccovar}{\dccut{\dcValue}{\dcfst{\dccovar}}}}
  \equiv
  \fst{\gafu{\dcValue}}$
  by Lemma~\ref{lem:subst_trans_inv} (1).
  Hence we have 
  \begin{align*}
    \gafu{
      \dcpairr{
        \dcabsr{\dccovar}{\dccut{\dcValue}{\dcfst{\dccovar}}}
      }{
        \dcabsr{\dccovar}{\dccut{\dcValue}{\dcsnd{\dccovar}}}
      }
    }
    &\cbveq
    \pair{
      \fst{\gafu{\dcValue}}
    }{
      \snd{\gafu{\dcValue}}
    }
    \cbveq
    \gafu{\dcValue}
    \enspace.
  \end{align*}

  The case of $(\eta\vee)$.
  Note that
  $\gafu{\dcabsl{\dcvar}{\dccut{\dcinl{\dcvar}}{\dccoterm}}} \cbveq \fst{\gafu{\dccoterm}}$.
  Hence we have 
  \begin{align*}
    \gafu{
      \dcpairl{
        \dcabsl{\dcvar}{\dccut{\dcinl{\dcvar}}{\dccoterm}}
      }{
        \dcabsl{\dcvar}{\dccut{\dcinr{\dcvar}}{\dccoterm}}
      }
    }
    &\cbveq
    \pair{
      \fst{\gafu{\dccoterm}}
    }{
      \snd{\gafu{\dccoterm}}
    }
    \cbveq
    \gafu{\dccoterm}
    \enspace.
  \end{align*}
  
  The cases of $(\eta R)$ and $(\eta L)$ are shown immediately.
  
  The case of $(\zeta)$ is shown by using Lemma~\ref{lem:evalcontext_trans_inv}:
  \begin{align*}
    \gafu{\dccut{\dcEval{\dcterm}}{\dccoterm}}
    &\equiv
    \bracket{
      \gafu{\dcEval{\dcterm}}
    }{
      \gafu{\dccoterm}
    }
    \cbveq
    \bracket{
      \gafu{\dcEvalSymb}\{\gafu{\dcterm}\}
    }{
      \gafu{\dccoterm}
    }
    \\
    &\cbveq
    \bracket{
      \gafu{\dcterm}
    }{
      \muabs{\evar}{
        \bracket{
          \gafu{\dcEvalSymb}\{\evar\}
        }{
          \gafu{\dccoterm}
        }
      }
    }
    \\
    &\cbveq
    \gafu{
      \dccut{\dcterm}{
        \dcabsl{
          \dcvar
        }{
          \dccut{
            \dcEval{\dcvar}
          }{
            \dccoterm
          }
        }
      }
    }
    \enspace.
  \end{align*}
\end{proof}


The equivalence between $\cbvSubDCeq$ and $\cbvBLCeq$ clarifies an essential difference
between the \emph{full} dual calculus and the bilateral $\lambda$-calculus.
The negation of the dual calculus is not involutive, 
since $\negop{\negop{A}}$ is not isomorphic to $A$. 
The dual calculus actually contains the involutive duality not as the object-level negation
but as the meta-level operation such as the antecedent and succedent duality of the sequent calculus.
On the other hand, the negation is represented using inversions of polarities in the bilateral $\lambda$-calculus.
By definition, the negation is involutive.

\section{Proofs}\label{sec:proofs}

\def\proofname{Proof of Proposition~\ref{prop:plusminus}}
\begin{proof}
  The proposition holds as follows:
  \begin{center}
    \AxiomC{$\plusop{A_0 \to A_1}$}
    \AxiomC{$[\plusop{A_0}]$}
    \BinaryInfC{$\plusop{A_1}$}
    \AxiomC{$[\minusop{A_1}]$}
    \BinaryInfC{$\bot$}
    \UnaryInfC{$\minusop{A_0}$}
    \UnaryInfC{$\minusop{A_0 \gets A_1}$}
    \bottomAlignProof
    \DisplayProof
    \qquad\qquad\qquad
    \AxiomC{$\minusop{A_0 \to A_1}$}
    \UnaryInfC{$\plusop{A_0}$}
    \AxiomC{$\minusop{A_0 \to A_1}$}
    \UnaryInfC{$\minusop{A_1}$}
    \BinaryInfC{$\plusop{A_0 \gets A_1}$}
    \bottomAlignProof
    \DisplayProof
  \end{center}
  \vspace{\baselineskip}
  \begin{center}
    \AxiomC{$\plusop{A_0 \gets A_1}$}
    \UnaryInfC{$\plusop{A_0}$}
    \AxiomC{$\plusop{A_0 \gets A_1}$}
    \UnaryInfC{$\minusop{A_1}$}
    \BinaryInfC{$\minusop{A_0 \to A_1}$}
    \bottomAlignProof
    \DisplayProof
    \qquad\qquad\qquad
    \AxiomC{$[\plusop{A_0}]$}
    \AxiomC{$\minusop{A_0 \gets A_1}$}
    \AxiomC{$[\minusop{A_1}]$}
    \BinaryInfC{$\minusop{A_0}$}
    \BinaryInfC{$\bot$}
    \UnaryInfC{$\plusop{A_1}$}
    \UnaryInfC{$\plusop{A_0 \to A_1}$}
    \bottomAlignProof
    \DisplayProof
  \end{center}
\end{proof}

\def\proofname{Proof of Proposition~\ref{prop:redundant}}
\begin{proof}
  1) It suffices to use $\Noncontradict$ and $\Reductio{}$ with
  $\RuleLabel{+}{\wedge}{E}{0}$,
  $\RuleLabel{+}{\wedge}{E}{1}$,
  $\RuleLabel{+}{\wedge}{I}{}$,
  $\RuleLabel{-}{\vee}{E}{}$,
  $\RuleLabel{-}{\vee}{I}{0}$, and
  $\RuleLabel{-}{\vee}{I}{1}$
  as follows:
  \begin{center}
  \AxiomC{$\minusop{A_i}$}
  \AxiomC{$[\plusop{A_0 \wedge A_1}]$}
  \UnaryInfC{$\plusop{A_i}$}
  \BinaryInfC{$\bot$}
  \UnaryInfC{$\minusop{A_0 \wedge A_1}$}
  \bottomAlignProof
  \DisplayProof
  \qquad\qquad
  \AxiomC{$\minusop{A_0 \wedge A_1}$}
  \AxiomC{$[\minusop{A_0}]$}
  \noLine
  \UnaryInfC{$\natvdots$}
  \noLine
  \UnaryInfC{$\mathcal{A}$}
  \AxiomC{$[\mathcal{A}^\ast]$}
  \BinaryInfC{$\bot$}
  \UnaryInfC{$\plusop{A_0}$}
  \AxiomC{$[\minusop{A_1}]$}
  \noLine
  \UnaryInfC{$\natvdots$}
  \noLine
  \UnaryInfC{$\mathcal{A}$}
  \AxiomC{$[\mathcal{A}^\ast]$}
  \BinaryInfC{$\bot$}
  \UnaryInfC{$\plusop{A_1}$}
  \BinaryInfC{$\plusop{A_0 \wedge A_1}$}
  \BinaryInfC{$\bot$}
  \UnaryInfC{$\mathcal{A}$}
  \bottomAlignProof
  \DisplayProof
  \end{center}
  \begin{center}
  \AxiomC{$\plusop{A_0 \vee A_1}$}
  \AxiomC{$[\plusop{A_0}]$}
  \noLine
  \UnaryInfC{$\natvdots$}
  \noLine
  \UnaryInfC{$\mathcal{A}$}
  \AxiomC{$[\mathcal{A}^\ast]$}
  \BinaryInfC{$\bot$}
  \UnaryInfC{$\minusop{A_0}$}
  \AxiomC{$[\plusop{A_1}]$}
  \noLine
  \UnaryInfC{$\natvdots$}
  \noLine
  \UnaryInfC{$\mathcal{A}$}
  \AxiomC{$[\mathcal{A}^\ast]$}
  \BinaryInfC{$\bot$}
  \UnaryInfC{$\minusop{A_1}$}
  \BinaryInfC{$\minusop{A_0 \vee A_1}$}
  \BinaryInfC{$\bot$}
  \UnaryInfC{$\mathcal{A}$}
  \bottomAlignProof
  \DisplayProof
  \qquad\qquad
  \AxiomC{$\plusop{A_i}$}
  \AxiomC{$[\minusop{A_0 \vee A_1}]$}
  \UnaryInfC{$\minusop{A_i}$}
  \BinaryInfC{$\bot$}
  \UnaryInfC{$\plusop{A_0 \vee A_1}$}
  \bottomAlignProof
  \DisplayProof
  \end{center}
where $i = 0,1$.

\noindent
  2) It suffices to use $\Noncontradict$ and $\Reductio{}$ with
  $\RuleLabel{+}{\to}{E}{}$,
  $\RuleLabel{+}{\to}{I}{}$,
  $\RuleLabel{-}{\gets}{E}{}$, and
  $\RuleLabel{-}{\gets}{I}{}$
  as follows:
  \begin{center}
    \AxiomC{$[\plusop{A_0 \to A_1}]$}
    \AxiomC{$\plusop{A_0}$}
    \BinaryInfC{$\plusop{A_1}$}
    \AxiomC{$\minusop{A_1}$}
    \BinaryInfC{$\bot$}
    \UnaryInfC{$\minusop{A_0 \to A_1}$}
    \bottomAlignProof
    \DisplayProof
    \qquad
    \AxiomC{$\minusop{A_0 \to A_1}$}
    \AxiomC{$[\plusop{A_0}]$}
    \AxiomC{$[\minusop{A_0}]$}
    \BinaryInfC{$\bot$}
    \UnaryInfC{$\plusop{A_1}$}
    \UnaryInfC{$\plusop{A_0 \to A_1}$}
    \BinaryInfC{$\bot$}
    \UnaryInfC{$\plusop{A_0}$}
    \bottomAlignProof
    \DisplayProof
  \end{center}
  \vspace{\baselineskip}
  \begin{center}
    \AxiomC{$\minusop{A_0 \to A_1}$}
    \AxiomC{$[\plusop{A_1}]$}
    \UnaryInfC{$\plusop{A_0 \to A_1}$}
    \BinaryInfC{$\bot$}
    \UnaryInfC{$\minusop{A_1}$}
    \bottomAlignProof
    \DisplayProof
    \qquad
    \AxiomC{$[\minusop{A_0 \gets A_1}]$}
    \AxiomC{$\minusop{A_1}$}
    \BinaryInfC{$\minusop{A_0}$}
    \AxiomC{$\plusop{A_0}$}
    \BinaryInfC{$\bot$}
    \UnaryInfC{$\plusop{A_0 \gets A_1}$}
    \bottomAlignProof
    \DisplayProof
  \end{center}
  \vspace{\baselineskip}
  \begin{center}
    \AxiomC{$\plusop{A_0 \gets A_1}$}
    \AxiomC{$[\minusop{A_0}]$}
    \UnaryInfC{$\minusop{A_0 \gets A_1}$}
    \BinaryInfC{$\bot$}
    \UnaryInfC{$\plusop{A_0}$}
    \bottomAlignProof
    \DisplayProof
    \qquad
    \AxiomC{$\plusop{A_0 \gets A_1}$}
    \AxiomC{$[\plusop{A_1}]$}
    \AxiomC{$[\minusop{A_1}]$}
    \BinaryInfC{$\bot$}
    \UnaryInfC{$\minusop{A_0}$}
    \UnaryInfC{$\minusop{A_0 \gets A_1}$}
    \BinaryInfC{$\bot$}
    \UnaryInfC{$\minusop{A_1}$}
    \bottomAlignProof
    \DisplayProof
  \end{center}
\end{proof}

\def\proofname{Proof of Theorem~\ref{thm:trans}}
\begin{proof}
The proposition holds because the following:
\begin{center}
    \AxiomC{$\plusj{\varPi'; \varSigma'}{\expr}{A_0 \to A_1}$}
    \AxiomC{$\plusj{\varPi'; \varSigma'}{\evar}{A_0}$}
    \BinaryInfC{$\plusj{\varPi'; \varSigma'}{\app{\expr}{\evar}}{A_1}$}
    \AxiomC{$\minusj{\varPi'; \varSigma'}{\cvar}{A_1}$}
    \BinaryInfC{$\zeroj{\varPi, \evar \colon A_0; \varSigma, \cvar \colon A_1}{\bracket{\app{\expr}{\evar}}{\cvar}}$}
    \UnaryInfC{$\minusj{\varPi; \varSigma, \cvar \colon A_1}{\muabs{\evar}{\bracket{\app{\expr}{\evar}}{\cvar}}}{A_0}$}
    \UnaryInfC{$\minusj{\varPi}{\lamabs{\cvar}{\muabs{\evar}{\bracket{\app{\expr}{\evar}}{\cvar}}}}{A_0 \gets A_1}$}
    \DisplayProof
\end{center}
  \vspace{\baselineskip}
\begin{center}
    \AxiomC{$\plusj{\varPi'; \varSigma'}{\evar}{A_0}$}
    \AxiomC{$\minusj{\varPi'; \varSigma'}{\cont}{A_0 \gets A_1}$}
    \AxiomC{$\minusj{\varPi'; \varSigma'}{\cvar}{A_1}$}
    \BinaryInfC{$\minusj{\varPi'; \varSigma'}{\app{\cont}{\cvar}}{A_0}$}
    \BinaryInfC{$\zeroj{\varPi, \evar \colon A_0; \varSigma, \cvar \colon A_1}{\bracket{\evar}{\app{\cont}{\cvar}}}$}
    \UnaryInfC{$\plusj{\varPi, \evar \colon A_0; \varSigma}{\muabs{\cvar}{\bracket{\evar}{\app{\cont}{\cvar}}}}{A_1}$}
    \UnaryInfC{$\plusj{\varPi; \varSigma}{\lamabs{\evar}{\muabs{\cvar}{\bracket{\evar}{\app{\cont}{\cvar}}}}}{A_0 \to A_1}$}
    \DisplayProof
\end{center}
are derived where $\varPi'$ and $\varSigma'$ are $\varPi, \evar
\colon A_0; \varSigma, \cvar \colon A_1$,
respectively.
\end{proof}

\end{document}